\documentclass[]{llncs}
\usepackage[utf8]{inputenc}
\usepackage{amsmath,amssymb}
\usepackage{algorithm}
\usepackage{algpseudocode}
\usepackage{url,ifthen}
\usepackage[T1]{fontenc}
\usepackage{tabularx}
\usepackage{graphicx}
\usepackage{listings}
\usepackage{enumerate}
\usepackage{float}
\usepackage{imakeidx}
\usepackage{multirow}
\usepackage[pagebackref=true]{hyperref}
\usepackage[left=2.5cm,right=2.5cm,top=2.5cm,bottom=2.5cm]{geometry}
\usepackage[toc,page]{appendix}
\usepackage{xcolor}

\usepackage{tikz}
\usetikzlibrary{arrows,positioning, calc}
\tikzstyle{vertex}=[draw,fill=black!15,circle,minimum size=20pt,inner sep=0pt]

\DeclareMathOperator{\wf}{\mathsf{WF}}

\DeclareMathOperator{\scc}{\mathsf{SC}}
\DeclareMathOperator{\hsc}{\mathsf{HSC}}
\DeclareMathOperator{\sctkem}{\mathsf{SCTKEM}}
\DeclareMathOperator{\kem}{\mathsf{KEM}}
\DeclareMathOperator{\dem}{\mathsf{DEM}}
\DeclareMathOperator{\cma}{\mathsf{CMA}}
\DeclareMathOperator{\sd}{\mathsf{SD}}

\DeclareMathOperator{\cca}{\mathsf{CCA2}}
\DeclareMathOperator{\euf}{\mathsf{EUF}}
\DeclareMathOperator{\pr}{\mathsf{Pr}}
\DeclareMathOperator{\oracle}{\mathsf{Oracles}}
\DeclareMathOperator{\sufcma}{\mathsf{SUF-CMA}}
\DeclareMathOperator{\eufcma}{\mathsf{EUF-CMA}}
\DeclareMathOperator{\indcca}{\mathsf{IND-CCA2}}
\DeclareMathOperator{\pregame}{\mathsf{Preimage.Game}}
\DeclareMathOperator{\symencrypt}{\mathsf{SymEncrypt}}
\DeclareMathOperator{\symdecrypt}{\mathsf{SymDecrypt}}
\DeclareMathOperator{\decode}{\mathsf{Decode}}
\DeclareMathOperator{\setup}{\mathsf{Setup}}
\DeclareMathOperator{\keygen}{\mathsf{KeyGen}}
\DeclareMathOperator{\encrypt}{\mathsf{Encrypt}}
\DeclareMathOperator{\decrypt}{\mathsf{Decrypt}}
\DeclareMathOperator{\sign}{\mathsf{Sign}}
\DeclareMathOperator{\sig}{\mathsf{sign}}
\DeclareMathOperator{\verif}{\mathsf{Verif}}
\DeclareMathOperator{\sym}{\mathsf{Sym}}
\DeclareMathOperator{\encap}{\mathsf{Encap}}
\DeclareMathOperator{\decap}{\mathsf{Decap}}
\DeclareMathOperator{\signcrypt}{\mathsf{Signcrypt}}
\DeclareMathOperator{\unsigncrypt}{\mathsf{Unsigncrypt}}
\DeclareMathOperator{\adv}{\mathsf{Adv}}
\DeclareMathOperator{\wt}{\mathsf{wt}}
\DeclareMathOperator{\pkegame}{\mathsf{PKE.Game}}
\DeclareMathOperator{\pkeenc}{\mathsf{PKE.Encrypt}}
\DeclareMathOperator{\pkedec}{\mathsf{PKE.Decrypt}}

\DeclareMathOperator{\pate}{\mathsf{PaTe}}
\DeclareMathOperator{\attack}{\mathsf{Attack}}

\DeclareMathOperator{\sk}{\mathsf{sk}}
\DeclareMathOperator{\pk}{\mathsf{pk}}

\DeclareMathOperator{\state}{\mathsf{state}}

\newcommand{\comment}[1]{}

\def\DHLhksqrt#1#2{\setbox0=\hbox{$#1\sqrt{#2\,}$}\dimen0=\ht0
\advance\dimen0-0.2\ht0
\setbox2=\hbox{\vrule height\ht0 depth -\dimen0}%
{\box0\lower0.4pt\box2}}

\newboolean{visible}
\newcounter{todoCounter}
\newcommand{\todo}[1]{%
\addtocounter{todoCounter}{1}%
\ifthenelse{\boolean{visible}}%
{\textsuperscript{\fcolorbox{red}{white}{\tiny \arabic{todoCounter}}}%
\marginpar{\fcolorbox{red}{white}{\parbox{0.9in}{
\arabic{todoCounter}: \color{red} \small \sf  #1 }}}}{} }
\setboolean{visible}{true}
\newboolean{anonymous}
\setboolean{anonymous}{false}
\begin{document}
\pagestyle{plain}

\title{A Code-based Hybrid Signcryption Scheme}

{ 
	\author{Jean Belo Klamti \and M. Anwar Hasan }

} 
{
\institute{Department of Electrical and Computer Engineering\\ University of Waterloo \\200 University Ave W, Waterloo,
ON, N2L 3G1\\
}
\date{\today}

\maketitle

\begin{abstract}
A key encapsulation mechanism ($\kem$) that takes as input an arbitrary string, i.e., a tag, is known as tag-$\kem$, while a scheme that combines signature and encryption is called signcryption. In this paper, we present a code-based signcryption tag-$\kem$  scheme. We utilize a code-based signature and an $\indcca$ (adaptive chosen ciphertext attack) secure version of McEliece's {encryption} scheme. The proposed scheme uses an equivalent subcode as a public code for the receiver, making the NP-completeness of the  subcode equivalence problem be one of our main security assumptions. We then base the signcryption tag-$\kem$  to design a code-based hybrid signcryption scheme. A hybrid scheme deploys asymmetric- as well as symmetric-key encryption. We give security analyses of both our schemes in the standard model and prove that they are secure against $\indcca$ (indistinguishability under adaptive chosen ciphertext attack) and $\sufcma$ (strong existential unforgeability under chosen message attack).

\end{abstract}

\section{Introduction}

In  public-key cryptography, the authentication and confidentiality of communication between a sender and a receiver are ensured  by a two-step approach called \textit{signature-then-encryption}. In this approach, the sender  uses a digital signature scheme to sign a message and then encrypt it using an encryption algorithm. 
The cost of delivering a message in a secure and authenticated way using the signature-then-encryption approach is essentially the sum of the cost of a digital signature and that of  encryption. 

In 1997, Y. Zheng introduced a new cryptographic primitive called \textit{signcryption} to provide both authentication and confidentiality  in a single logical step \cite{Zh97}. In general, one can expect the cost of signcryption to be noticeably less than that of signature-then-encryption.
Zheng's sincryption scheme is based on the hardness of the discrete logarithm problem. Since Zheng's work, a number of signcryption schemes based on different hard assumptions have been introduced, see for example \cite{Zh97,ZhIm98,StZh00,YaCaLiXu19,LiMuKhTa,BaLiMcQu10,BaLiBMcQu10.1,DeMa10,SaSh18,YaWaWaYaYa13,LeDuRoSuFuKi21,ZhWa14}. Of these, the most efficient ones have followed Zheng's approach, i.e., used symmetric-key encryption as a black-box component \cite{BaLiMcQu10,BaLiBMcQu10.1,DeMa10}. It has been of interest to many researchers to study how a combination of asymmetric- and symmetric-key encryption schemes could be used to build efficient signcryption schemes in a more general setting. 

To that end, Dent in 2004 proposed the first formal composition model for hybrid signcryption \cite{Dent04} and in 2005 developed an efficient model for signcryption $\kem$s in the \textit{outsider-} and the \textit{insider}-secure setting \cite{Dent05,Dent051}. In the outsider-secure setting the adversary is assumed to be distinct from the sender and receiver, while in the insider-secure setting the adversary is assumed to be a second party (i.e., either sender or receiver). 

In order to improve the model for the insider-secure setting in hybrid signcryption, Bj{\o}rstad and Dent  in 2006 proposed a model based on encryption tag-$\kem$ rather than regular encryption $\kem$ \cite{BjDe06}. Their model provides a simpler description of signcryption with a better generic security reduction for the  signcryption tag-$\kem$ construction. A year after Bj{\o}rstad and Dent's work, Yoshida and Fujiwara reported the first study of multi-user setting security of signcryption tag-$\kem$s \cite{YoFu07} which is a more suitable setting for the analysis of insider-secure schemes. 

\vspace{0.5em}
\subsubsection*{Motivation}

Most of the aforementioned signcryption schemes are based on the hardness of either the discrete logarithm  or the integer factorization problem and would be broken with the arrival of sufficiently large quantum computers. Therefore it is of interest to design signcryption schemes for the post-quantum era. Coding theory has some hard problems that are considered quantum-safe and in this paper, we explore the design of code-based signcryption. 

The first attempt for code-based signcryption was presented in 2012 by {Preetha et al.} \cite{MaVaRa12}. After that work, an attribute-based signcryption scheme  using linear codes was introduced in 2017 by {Song et al.} \cite{SoLiLiLi17}. Code-based signcryption remains an active area of research, specifically to study the design of cryptographic primitives like signcryption schemes that are quantum-safe.

\vspace{0.5em}
\subsubsection*{Contributions}

{In this paper, we present a signcryption tag-$\kem$  scheme using a probabilistic full domain hash (FDH) like code-based signature and a CCA2 secure version of  McEliece's encryption scheme. The underlying code-based signature in our scheme is called \textit{Wave}  introduced  by Debris-Alazard et al. \cite{Dags18}, while the CCA2 secure version of the McEliece scheme is based on the Fujisaki-Okamoto transformation introduced by Cayrel et al. \cite{CaHoPe12}. For the underlying McEliece scheme, we use a generator matrix of permuted Goppa subcodes as receivers' public keys. {With this feature, we are able to reduce the public key size of our scheme and include the subcode equivalence problem as one of your security assumptions. Because of the latter, for the key recovery attack, even if an adversary is able to distinguish whether the underlying code is a Goppa code, it has to solve the subcode equivalence problem which is NP-complete.} Thus, with well-chosen parameters, the most efficient attack against our scheme will be a brute-force attack.}

{Based on the signcryption tag-$\kem$, we design a code-based hybrid signcryption scheme. Then we give  security analyses of these two schemes in the standard model  assuming the insider-secure setting. Finally, we give a comparison of the hybrid signcryption with some relevant lattice-based signcryptions  in terms of key and ciphertext sizes.}

\vspace{0.5em}
\subsubsection*{Organization}

This paper is organized as follows. In Section \ref{Pre}, we first recall some basic notions of coding theory {and then briefly describe  relevant encryption   and } signature schemes that are of interest to this work. Section \ref{HSigncrypt} has the definition {and framework of signcryption and hybrid signcryption, and a brief review of the relevant security model.} We present our {sigcryption and hybrid sigcryption schemes} in Section \ref{gen_Code_base_Signcrypt} and then provide security analyses of the proposed schemes in Section \ref{Sec:Analysis}. We provide a set of parameters for the {hybrid sigcryption} scheme in Section \ref{Param} and then conclude in Section \ref{Conclusion}.

\vspace{0.5em}
\subsubsection*{Notations}

In this paper we use the following  notations:
\begin{itemize}
\item $\mathbb{F}_{q}$: finite field of size $q$ where $q=p^m$ is a prime power.
\item $\mathcal{C}$: $\mathbb{F}_{q}$-linear code of length $n$.
\item $\textbf{\textit{x}}$: a word or vector of $\mathbb{F}_{q}^n$.
\item $\wt(\textbf{\textit{x}})$: weight of $\textit{\textbf{x}}$.
\item $\mathbf{G}$ (resp. $\mathbf{H}$): generator (resp. parity-check) matrix of linear code $\mathcal{C}$.
\item $\mathcal{W}_{q,n, t}$ is the set of $q$-ary vectors of length $n$ and weight $t$.
\item $\sk_s$ (resp. $\sk_r$):  sender's (resp. receiver's) secrete key for  signcryption.
\item $\pk_s$ (resp. $\pk_r$): sender's (resp. receiver's) public key for  signcryption. 

\end{itemize}
\section{Preliminaries}\label{Pre}

In this section, we recall some notions pertaining to coding theory and code-based cryptography. 
\subsection{Coding theory and some relevant hard problems}$\ $

Let us consider the finite field $\mathbb{F}_{q}$. A $q$-ary linear code $\mathcal{C}$ of length $n$ and dimension $k$ over $\mathbb{F}_{q}$ is a vector subspace of dimension $k$ of $\mathbb{F}_{q}^{n}$. It can be specified by a full rank matrix $\mathbf{G}\in \mathbb{F}_{q}^{k\times n}$, called \textit{generator matrix} of $\mathcal{C}$, whose rows span the code. Namely, $\mathcal{C}=\left\lbrace \textit{\textbf{x}}\mathbf{G} \ \text{s.t.} \ \textit{\textbf{x}}\in \mathbb{F}_{q}^{k} \right\rbrace $. A linear code can  also be defined by the right kernel of matrix $\mathbf{H}\in \mathbb{F}_{q}^{r\times n}$, called \textit{parity-check matrix} of $\mathcal{C}$, as follows:
$$
\mathcal{C}=\left\lbrace \textit{\textbf{x}}\in \mathbb{F}_{q}^{n}  \ \ \text{s.t.} \ \  \mathbf{H}\textit{\textbf{x}}^{T}=\textit{\textbf{0}} \right\rbrace 
$$

  The \textit{Hamming distance} between two codewords is the number of positions (coordinates) where they differ. The \textit{minimal} distance of a code is the minimal distance of all codewords.

  The \textit{weight} of a word or vector $\textit{\textbf{x}}\in \mathbb{F}_{q}^{n}$, denoted by $wt\left(\textit{\textbf{x}}\right) ,$ is the number of its nonzero positions. Then the \textit{minimal} weight of a code $\mathcal{C}$ is the minimal weight of all nonzero codewords. In the case of linear code $\mathcal{C}$, its minimal distance is equal to the minimal weight of the code.

Below we recall some hard  problems  that are relevant to our discussions and analyses presented in this article. 

\begin{problem}\label{SD}(Binary syndrome decoding (SD) problem)
Given a  
matrix $\mathbf{H}\in \mathbb{F}_2^{r\times n}$, a vector $\textit{\textbf{s}}\in \mathbb{F}_2^r$, and an integer $\omega > 0$, 
find a 
vector $\textit{\textbf{y}}\in \mathbb{F}_2^n$ such that $\wt(\textit{\textbf{y}})=\omega$ and $\textit{\textbf{s}}=\textit{\textbf{y}}\mathbf{H}^T$.
\end{problem}

The syndrome decoding problem was proven to be NP-complete in 1978 by Berlekamp et al. \cite{BeMcVa78}. It is equivalent to the following problem.

\begin{problem}\label{GBD}(General decoding (GD) problem)
Given a 
matrix $\mathbf{G}\in \mathbb{F}_2^{k\times n}$, a vector  $\textit{\textbf{y}}  \in \mathbb{F}_2^n$, and an integer $\omega > 0$, 
find two 
vectors $\textit{\textbf{m}}\in \mathbb{F}_q^{k}$ and $\textit{\textbf{e}}\in \mathbb{F}_q^n$ such that $\wt(\textit{\textbf{e}})=\omega$ and $\textit{\textbf{y}}=\textit{\textbf{m}}\mathbf{G} \oplus \textit{\textbf{e}}$.
\end{problem}

The following problem is used in the security proof of the underlying signature that we use in this paper. It was first considered by Johansson and Jonsson in \cite{JoJo02}. It was analyzed later by Sendrier in \cite{Sen11}. 

\begin{problem}(Decoding One Out of Many (DOOM) problem)\label{DOOM}
Given a 
matrix $\mathbf{H}\in \mathbb{F}_q^{r\times n}$, a set of vector $\textit{\textbf{s}}_1$, $\textit{\textbf{s}}_2$,...,$\textit{\textbf{s}}_N\in \mathbb{F}_q^{r}$ and an integer $\omega$, 
find a 
vector $\textit{\textbf{e}}\in \mathbb{F}_q^n$ and an integer $i$ such that $1\leq i\leq N$, $\wt(\textit{\textbf{e}})=\omega$ and $\textit{\textbf{s}}_i=\textit{\textbf{e}}\mathbf{H}^T$.
\end{problem}

\begin{problem} \label{pr:distinguish} (Goppa code distinguishing (GCD) problem)
{Given a 
matrix $\mathbf{G}\in \mathbb{F}_2^{k\times n}$, 
decide 
whether $\mathbf{G}$ is a random binary or generator matrix of a Goppa code.}
\end{problem}

Faugère et al. \cite{FaGaOtPeTi13} showed that  Problem \ref{pr:distinguish} can be solved in special cases of Goppa codes with high rate.



The following is one of the problems, which the security assumption of our scheme's underlying signature mechanism relies on.

\begin{problem}(Generalized (${U}, {U}+{V}$)  code distinguishing problem.)\label{UVDP}
Given a 
matrix $\mathbf{H}\in \mathbb{F}_q^{r\times n}$, 
decide 
whether $\mathbf{H}$ is a parity-check matrix of a generalized ($U, U+V$)-code.
\end{problem}

Problem \ref{UVDP} was shown to be hard in the worst case by 
Debris-Alazard et al. \cite{DeSeTi17}
since it is NP-complete. Below, we recall the  subcode equivalence problem which is one of the problems on which the security assumption of our scheme is based. This problem  was proven to be NP-complete in 2017 by Berger et al. \cite{BGK17}. 

\begin{problem}(Subcode equivalence problem \cite{BGK17})
Given two linear codes $\mathcal{C}$ and $\mathcal{D}$ of length $n$ and respective dimension $k$ and $k'$, $k' \leq  k$, over the same finite field  $\mathbb{F}_{q}$, determine whether there exists a permutation $\sigma$ of  the support such that $\sigma(\mathcal{C})$  is a subcode of $\mathcal{D}$.
\end{problem}

\subsection{Code-based encryption }\label{Code-based}

The first code-based encryption was introduced in 1978 by R. McEliece \cite{Mce}. Below (in Figure \ref{fig:McElieceFO}) we give the McEliece scheme  Fujisaki-Okamoto transformation \cite{CaHoPe12} which comprises three algorithms: key generation, encryption, and decryption. 

\begin{figure}[ht]
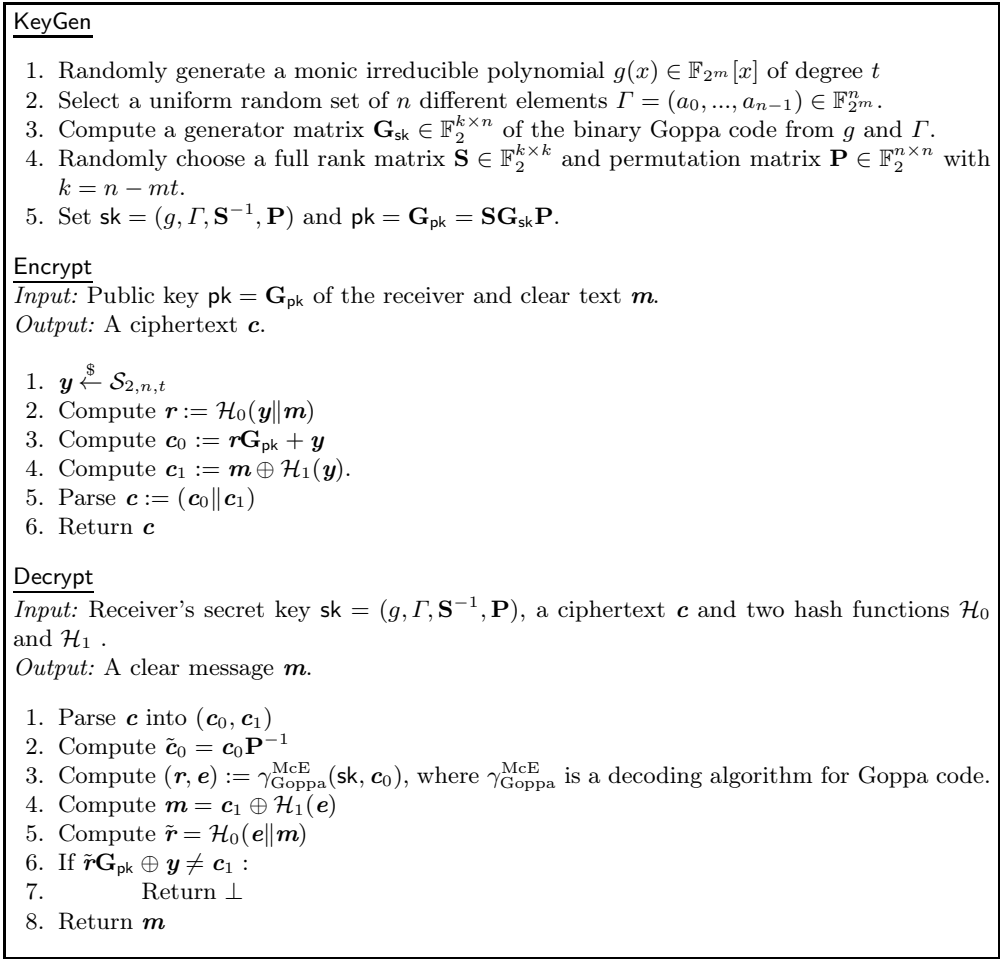

	\begin{center}
		\framebox{\parbox{13cm}{
	{\em \underline{$\keygen$}} 
	\begin{enumerate}
\item Randomly generate a monic irreducible polynomial $g(x)\in \mathbb{F}_{2^m}[x]$ of degree $t$
\item Select a uniform random set of $n$ different elements $\Gamma=(a_0,...,a_{n-1})\in \mathbb{F}_{2^m}^{n}$.
\item Compute a generator matrix $\mathbf{G}_{\sk}\in \mathbb{F}_2^{k\times n}$ of the binary Goppa code from $g$ and $\Gamma$.
\item Randomly choose a full rank matrix $\mathbf{S}\in \mathbb{F}_2^{k\times k}$ and permutation matrix $\mathbf{P}\in \mathbb{F}_2^{n\times n}$ with $k=n-mt$.
\item Set $\sk=(g, \Gamma, \mathbf{S}^{-1}, \mathbf{P})$ and  $\pk=\mathbf{G}_{\pk}=\mathbf{S}\mathbf{G}_{\sk}\mathbf{P}$.
\end{enumerate}
	\noindent\textbf{\underline{$\encrypt$}}

\noindent\textit{Input:} Public key $\pk=\mathbf{G}_{\pk}$ of the receiver and clear text $\textit{\textbf{m}}$. 
 
\noindent\textit{Output:} A ciphertext $\textit{\textbf{c}}$.
\begin{enumerate} 
\item $\textit{\textbf{y}}\stackrel{\$}{\leftarrow} \mathcal{S}_{2,n, t}$
\item Compute $\textit{\textbf{r}}:=\mathcal{H}_0(\textit{\textbf{y}}\Vert \textit{\textbf{m}})$ 
\item Compute $\textit{\textbf{c}}_0:=\textit{\textbf{r}}\mathbf{G}_{\pk} +\textit{\textbf{y}}$
\item Compute $\textit{\textbf{c}}_1:=\textit{\textbf{m}}\oplus\mathcal{H}_1(\textit{\textbf{y}})$.
\item Parse $\textit{\textbf{c}}:=(\textit{\textbf{c}}_0\Vert \textit{\textbf{c}}_1)$  
\item Return $\textit{\textbf{c}}$
\end{enumerate} 

\noindent\textbf{\underline{$\decrypt$}}

\noindent\textit{Input:} Receiver's secret key $\sk=(g, \Gamma, \mathbf{S}^{-1},  \mathbf{P})$,  a ciphertext $\textit{\textbf{c}}$ and two hash functions $\mathcal{H}_0$ and $\mathcal{H}_1$ . 
 
\noindent\textit{Output:} A clear message $\textit{\textbf{m}}$.
\begin{enumerate}
\item Parse $\textit{\textbf{c}}$ into ($\textit{\textbf{c}}_0,\textit{\textbf{c}}_1$) 
\item Compute $\tilde{\textit{\textbf{c}}}_0=\textit{\textbf{c}}_0\mathbf{P}^{-1}$
\item Compute $(\textit{\textbf{r}}, \textit{\textbf{e}}):=\gamma_{\text{Goppa}}^{\text{McE}}(\sk, {\textit{\textbf{c}}}_0)$, where $\gamma_{\text{Goppa}}^{\text{McE}}$ is a decoding algorithm for Goppa code.
\item Compute ${\textit{\textbf{m}}}= \textit{\textbf{c}}_1 \oplus \mathcal{H}_1(\textit{\textbf{e}})$
\item Compute $\tilde{\textit{\textbf{r}}}=\mathcal{H}_0(\textit{\textbf{e}}\Vert \textit{\textbf{m}})$
\item If $\tilde{\textit{\textbf{r}}}\mathbf{G}_{\pk}\oplus \textit{\textbf{y}} \neq \textit{\textbf{c}}_1:$
\item \hspace{1cm} Return $\perp$
\item  Return $ \textit{\textbf{m}}$
\end{enumerate}
}}  
\end{center}
\caption{\textit{McEliece's scheme with Fujisaki-Okamoto transformation}}
		\label{fig:McElieceFO}
	\end{figure}

The main drawback of the {McEliece} encryption scheme is its very large key size. To address this issue, many variants  of McEliece's scheme have been proposed, see for example \cite{Berg05,BCG09,MisBar09,MisTiSenBar2013,BarLinMiso2011,Persi2012}. In order to reduce the size of both public and private keys in code-based cryptography,  {H. Niederreiter} in 1986 introduced  a new cryptosystem \cite{Nied1986}. Niederreiter's cryptosystem is a dual version of McEliece's cryptosystem with some additional properties such that the ciphertext length  is relatively smaller. Indeed, the public key in Niederreiter's cryptosystem is a parity-check matrix instead of a generator matrix. In addition, ciphertexts are syndrome vectors instead of erroneous codewords. However, the {McEliece} and the Niederreiter schemes are equivalent from the security point of view due to the fact that Problems \ref{SD} and \ref{GBD} are equivalent.
%
\vspace{0.5cm}

\noindent \textbf{Code-based hybrid encryption:} A hybrid encryption scheme is a cryptographic protocol that features both an asymmetric-  and a symmetric-key encryption scheme. The first component is known as Key Encapsulation Mechanism ($\kem$), while the second is called Data Encapsulation Mechanism ($\dem$). The framework was first introduced in 2003 by Cramer and Shoup \cite{CrSh03} and later  the first code-based hybrid encryption was introduced in 2013 by Persichetti \cite{Persichetti13} using  {Niederreiter}'s encryption scheme. Persichetti's scheme was implemented in 2017 by {Cayrel et al.}  \cite{CaGuNdPe17}. After {Persichetti}'s work, some other code-based hybrid encryption schemes have been reported, e.g.,
\cite{MaVaRa13}.

\subsection{Code-based signature}
{Designing a secure and practical code-based signature scheme is still an open problem. The first secure code-based signature scheme was introduced by Courtois et al. (CFS) \cite{CFS01}. It is a full domain hash (FDH) like signature with two 
security assumptions: the indistinguishability of random binary linear codes and the hardness of syndrome decoding problem. To address some of the drawbacks of Courtois et al.'s scheme, Dallot proposed a modified  version, called mCFS, which is provably secure. Unfortunately, this scheme is not practical due to the difficulties  of finding a random decodable syndrome. In addition, 
the assumption of the indistinguishability of random binary Goppa codes has led to the emergence of attacks as described 
in \cite{FaGaOtPeTi13}. One of the latest code-based signature
schemes of this type is called Wave \cite{DeSeTi18}. It 
is based on generalized ($U, U+V$)-codes. It is secure and more efficient than the CFS signature scheme. In addition, it has a smaller signature size than almost all finalist candidates in the NIST post-quantum cryptography standardization process \cite{BaDeNe21}.}

Apart from the full domain hash approach, it is possible to design signature schemes by applying the Fiat-Shamir transformation \cite{Fiat86} to an identification protocol. To this end, one may use a code-based identification scheme like that of Stern \cite{Stern93}, Jain et al. \cite{JaKrPiTe12},  or Cayrel et al. 
\cite{CVE10}. This approach however leads to a signature scheme with a very large signature size. To address this issue, Lyubashevsky's framework \cite{Lyub09} can apparently be adapted. Unfortunately, almost all code-based signature schemes in Hamming  metric designed by using this framework have been cryptanalyzed \cite{BiMiPeSa20,Persichetti2018,Persichetti12,RaCoSS17,LiXiYe20,SoHuMuWuHu20}. The only one which has remained secure so far is a rank metric-based signature scheme proposed by Aragon et al.\cite{Durendal18}.

In Figure \ref{fig:Wave}, we recall Debris-Alazard et al.'s signature scheme (Wave) which is of our interest for this work. In Wave, the secret key is a tuple of three matrices $\sk=(\mathbf{S}, \mathbf{H}_{\sk}, \mathbf{P})$, where $\mathbf{S}\in \mathbb{F}_q^{r\times r}$ is an invertible matrix, $\mathbf{H}_{\sk}\in \mathbb{F}_q^{r\times n}$ is a parity-check matrix of a generalized ($U, U+V$)-code and $\mathbf{P}\in \mathbb{F}_2^{n\times n}$ is a permutation matrix. The public key is a matrix $pk=\mathbf{H}_{\pk}$, where $\mathbf{H}_{\pk}=\mathbf{S}\mathbf{H}_{\sk}\mathbf{P}$. Steps for signature and verification processes are given in  Figure \ref{fig:Wave}. For additional details, the reader is referred to \cite{DeSeTi18a,DeSeTi18}.

\begin{figure}[ht]
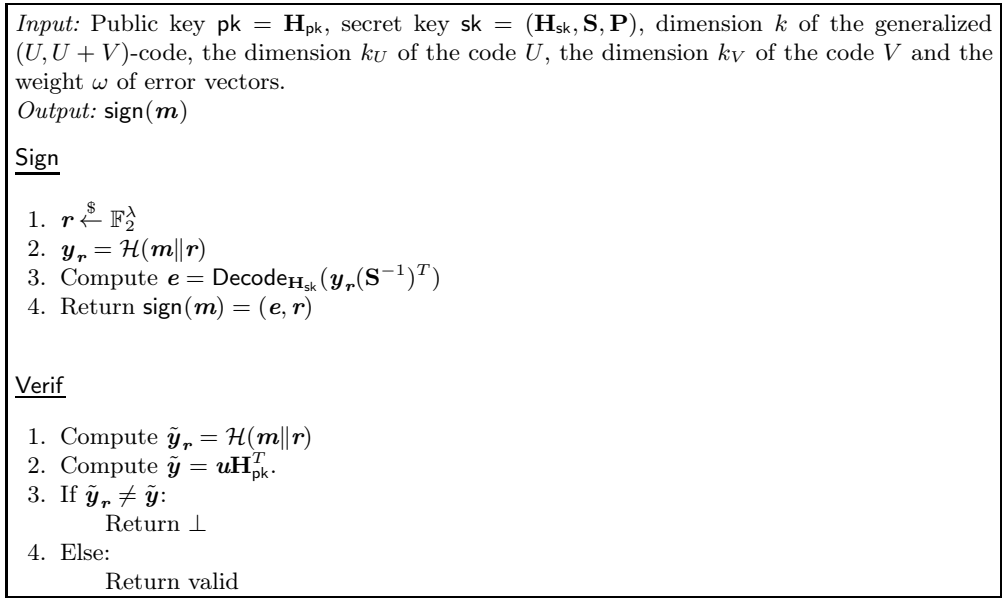

	\begin{center}
		\framebox{\parbox{13cm}{
		\begin{minipage}[top]{13cm}
		    {\em\textit{Input:}} Public key $\pk={\mathbf{H}_{\pk}}$, secret key $ \sk=(\mathbf{H}_{\sk}, \mathbf{S}, \mathbf{P})$, dimension $k$ of the generalized ($U, U+V$)-code, the dimension $k_U$ of the code $U$, the dimension $k_V$ of the code $V$ and the weight $\omega$ of error vectors.

            {\em\textit{Output:}} $\sig(\textbf{\textit{m}})$
		\end{minipage}
		
\vspace{0.3cm}

\begin{minipage}[left]{13cm}
 {\em\underline{$\sign$}} 
	\begin{enumerate}

		\item $\textit{\textbf{r}}\stackrel{\$}{\leftarrow} \mathbb{F}_{2}^{\lambda}$
		\item $\textbf{\textit{y}}_{\textit{\textbf{r}}}=\mathcal{H}( \textit{\textbf{m}} \Vert \textit{\textbf{r}} )$
		\item Compute $\textit{\textbf{e}}=\decode_{\mathbf{H}_{\sk}}( \textbf{\textit{y}}_{\textit{\textbf{r}}}(\mathbf{S}^{-1})^T)$ 
		\item Return $\sig(\textit{\textbf{m}})=(\textit{\textbf{e}},  \textit{\textbf{r}})$
		\item[] 
	\end{enumerate}

	{\em\underline{$\verif$}}
	\begin{enumerate}
		\item Compute $\tilde{\textbf{\textit{y}}}_{\textit{\textbf{r}}}=\mathcal{H}( \textit{\textbf{m}}\Vert \textit{\textbf{r}} )$
		\item Compute $\tilde{\textbf{\textit{y}}}= \textit{\textbf{u}}{\mathbf{H}_{\pk}^T}$.
		\item If $\tilde{\textbf{\textit{y}}}_{\textit{\textbf{r}}} \neq \tilde{\textbf{\textit{y}}}$:
		\begin{itemize}
		    \item[] Return $\perp$
		\end{itemize}
		\item Else:
		\begin{itemize}
		    \item[] Return valid
		\end{itemize}
	\end{enumerate}
\end{minipage}
}} 
		\end{center}
		\caption{\textit{Wave signature scheme }\cite{DeSeTi18}}
		\label{fig:Wave}
	\end{figure}

\section{Signcryption and security model}\label{HSigncrypt}

In this section, we first recall the definition of signcryption followed by the signcryption tag-$\kem$ framework  and its security model under the insider setting.

\subsection{Signcryption and its tag-KEM framework}
\textbf{Signcryption:} A signcryption scheme is a tuple of algorithms $\scc$=($\setup$, $\keygen_s$, $\keygen_r$, $\signcrypt$, $\unsigncrypt$) \cite{BaStZh07} where:
\begin{enumerate}
\item[$\ast$] $\setup$($1^{\lambda}$) is the common parameter generation algorithm  with $\lambda$, the security parameter,
\item[$\ast$] $\keygen_s$(resp. $\keygen_r$) is a key-pair generation algorithm for the sender (resp. receiver),
\item[$\ast$] $\signcrypt$ is the signcryption algorithm and
\item[$\ast$] $\unsigncrypt$ corresponds to the unsigncryption algorithm.
\end{enumerate}

 For more details on the design of signcryption, the reader is referred to \cite{YuDeZh} (Chap. 2, Sec. 3, p. 30).

\vspace{0.5cm}
\noindent
{\bf Signcryption tag-$\kem$:} A signcryption tag-$\kem$ denoted by $\sctkem$ is a tuple of algorithms \cite{BjDe06}: $$\sctkem=(\setup, \keygen_s, \keygen_r, \sym, \encap, \decap)$$ where,
\begin{itemize}
\item $\setup$ is an algorithm for generating common parameters. 
\item $\keygen_s$ (resp. $\keygen_r$) is the sender (resp. receiver) key generation algorithm. It takes as input the
global information $I$, and returns a private/public keypair ($\sk_s$, $\pk_s$) (resp. ($\sk_r$, $\pk_r$)) that is used to send signcrypted messages.
\item $\sym$ is a symmetric key generation algorithm. It takes as input the private key of the sender $\sk_s$ and the public key of the receiver $\pk_r$ and outputs a symmetric key $K$ together with internal state information $\varpi$.
\item $\encap$ takes as input the state information $\varpi$ together with an arbitrary string $\tau$, which is called a tag, and outputs an encapsulation $E$. 
\item $\decap$ is the decapsulation/verification algorithm. It takes as input the sender's public key $\pk_s$, the receiver's private key $\sk_r$, an encapsulation $E$, and a tag $\tau$. It returns either symmetric key $K$ or the unique error symbol $\perp$.
\end{itemize}

\vspace{0.5cm}
\noindent
{\bf Hybrid signcryption tag-$\kem$+$\dem$:} It is simply a combination of a $sctkem$ and a regular Data Encapsulation Mechanism ($\dem$).  
 
\subsection{Insider security for signcryption tag-KEM}

{\bf $\indcca$ game in signcryption tag-$\kem$:} It corresponds to a game between a challenger and a probabilistic polynomial-time adversary $\mathcal{A}_{\cca}$ such that  the latter tries to distinguish whether a given session key $K$ is the one embedded in an encapsulation or not. During this game, $\mathcal{A}_{\cca}$ has adaptive access to three oracles for the attacked user corresponding to 
algorithms
$\sym$, $\encap$, and $\decap$ \cite{BjDe06,YuDeZh,YoFu07}. The game is described in Figure \ref{CCA2game} below.

\begin{figure}[ht]
	\begin{center}
		\framebox{\parbox{13cm}{
	{\underline{$\oracle$}} 
	\begin{enumerate}
\item $\mathcal{O}_{\sym}$ is the symmetric key generation oracle with input a public key $\pk$, and computes ($K$, $\omega$) = $\sym$($\sk_s$, $\pk$). It then stores the value $\omega$ (hidden from the view of the adversary, and overwriting any previously stored values),
and returns the symmetric key $K$.
\item $\mathcal{O}_{\encap}$ is the key encapsulation oracle. It takes an arbitrary tag $\tau$ as input and checks whether there exists a stored value $\omega$. If there is not, it returns $\perp$
and terminates. Otherwise, it erases the value from storage and returns $E=\encap(\omega, \tau)$.
\item $\mathcal{O}_{\decap}$ corresponds to the decapsulation/verification oracle. It takes an encapsulation $E$, a tag $\tau$, any sender's public key $\pk$ as input and returns $\decap(\pk, \sk_r, E, \tau)$.
\end{enumerate}
	{\underline{$\indcca$ Game for $\sctkem$}} 
	\begin{enumerate}
	\item $I:=\setup(1^{\lambda})$ 
	\item $(\sk_r, \pk_r):= \keygen_r(I)$
	\item $(\sk_s, {\state}_1):=\mathcal{A}_{\cca}^{\mathcal{O}_{\sym},\mathcal{O}_{\encap}, \mathcal{O}_{\decap}}(\pk_r)$
	\item $(K_1, \varpi):=\sym(\sk_s, \pk_r)$, $K_0\stackrel{\$}{\leftarrow} \mathcal{K}$ and $b\stackrel{\$}{\leftarrow} \{0, 1\}$
	\item $(\tau, {\state}_2):=\mathcal{A}_{\cca}^{\mathcal{O}_{\sym},\mathcal{O}_{\encap}, \mathcal{O}_{\decap}}( K_b, {\state}_1)$
	\item $E:=\encap(\varpi, \tau)$
	\item $b':=\mathcal{A}_{\cca}^{\mathcal{O}_{\sym},\mathcal{O}_{\encap}, \mathcal{O}_{\decap}}(E,{\state}_2)$
	\end{enumerate}
    }}  
	\end{center}
	\caption{{$\indcca$ game} \cite{YoFu07}.}
	\label{CCA2game}
	\end{figure}

During  Step 7, the adversary $\mathcal{A}_{\cca}$ is restricted not to make  decapsulation queries on  $(E, \tau)$ to the decapsulation oracle. The advantage of the adversary $\mathcal{A}$ is defined by:
$$
\adv(\mathcal{A}_{\cca})=|\pr (b'=b)-1/2|.
$$ 

A  signcryption tag-$\kem$ is {$\indcca$} secure
if, for any adversary $\mathcal{A}$, its advantage in the {$\indcca$} game is negligible with respect to the security parameter $\lambda$.

\vspace{0.5cm}
 \noindent
{\bf $\sufcma$ game for signcryption tag-$\kem$:}
This game is a challenge between a challenger and a probabilistic polynomial-time adversary  (i.e., a forger) $\mathcal{F}_{\cma}$. In this game, the forger tries to generate a valid encapsulation $E$ from the sender to any receiver, with adaptive access to the three oracles. The adversary is allowed to come up with the presumed secret key $\sk_r$ as part of his forgery  \cite{YoFu07}:

\begin{figure}[ht]
	\begin{center}
		\framebox{\parbox{13cm}{
	{\underline{$\sufcma$ Game for $\sctkem$}} 
	\begin{enumerate}
	\item $I:=\setup(1^{\lambda})$ 
	\item $(\sk_s, \pk_s):= \keygen_r(I)$
	\item $(E, \tau, \sk_r):=\mathcal{F}_{\cma}^{\mathcal{O}_{\sym},\mathcal{O}_{\encap}, \mathcal{O}_{\decap}}(\pk_s)$
	\end{enumerate}
}}  
		\end{center}
		\caption{\textit{$\sufcma$ game} \cite{YoFu07}.}
		\label{SUFgame}
	\end{figure}
The adversary $\mathcal{F}_{\cma}$ wins the $\sufcma$ game if
$$\perp\neq  \decap(\pk_s, \sk_r, E, \tau)$$ and the encapsulation oracle never returns $E$ when he queries on the tag $\tau$. 
The advantage of $\mathcal{F}_{\cma}$ is the probability that $\mathcal{F}_{\cma}$ wins the {$\sufcma$} game. A  signcryption tag-$\kem$ is {$\sufcma$} secure if the winning probability of the {$\sufcma$} game by $\mathcal{F}_{\cma}$ is negligible.

\begin{definition}
A  signcryption tag-$\kem$ is said to be secure if it is $\indcca$ and $\sufcma$ secure.
\end{definition}

\subsection{Generic security criteria of hybrid signcryption tag-KEM+DEM}

\vspace{0.5cm}
 \noindent
{\bf Security criteria for hybrid signcryption:} The security of a hybrid signcryption tag-$\kem$+$\dem$ depends on those of the underlying signcryption tag-$\kem$ and $\dem$. However, it is important to note that in  the standard model a signcryption tag-$\kem$ is secure if it is both $\indcca$ and $\sufcma$ secure. Therefore, the generic security criteria for hybrid signcryption tag-$\kem$+$\dem$ is given by the following theorem: 

\begin{theorem}\label{Genesec} \cite{YoFu07,BjDe06}$\ $
Let $\hsc$ be a hybrid signcryption scheme constructed from a signcryption {{tag-$\kem$}} and a $\dem$. If the signcryption tag-$\kem$ is $\indcca$ secure and the $\dem$ is one-time secure, then $\hsc$ is $\indcca$ secure. Moreover, if the signcryption tag-$\kem$ is $\sufcma$ secure, then $\hsc$ is also $\sufcma$ secure.
\end{theorem}

\section{Code-based hybrid signcryption}\label{gen_Code_base_Signcrypt}

In this section, we first design a code-based signcryption tag-$\kem$  scheme. Then we combine it with a one-time (OT) 
secure $\dem$ for designing a hybrid signcryption tag-$\kem$+$\dem$ scheme.
 
\subsection{Code-based signcryption tag-KEM scheme}\label{tagKEM}

For  designing  our code-based signcryption tag-$\kem$ scheme, we use  the McEliece scheme as the underlying encryption scheme.  More specifically,  in order to achieve the $\indcca$ security  for our schemes, we use McEliece's  scheme with the Fujisaki-Okamoto transformation  \cite{FuOk99,CaHoPe12}. The authors of \cite{CaHoPe12} gave an instantiation of this scheme using generalized Srivastava (GS) codes. Indeed, by using GS codes, it seems possible to choose secure parameters even for codes defined over relatively small extension fields. However, Barelli and Couvreur recently introduced  an efficient structural attack  \cite{BaCo18} against some of the candidates in the NIST post-quantum cryptography standardization process. Their attack is against code-based encryption schemes using some quasi-dyadic alternant codes with extension degree $2$. It works specifically for schemes based on GS code called DAGS \cite{Dags18}. Therefore, in our work, we use the Goppa code with the Classic McEliece parameters. As for the underlying signature scheme, we use the code-based Wave \cite{DeSeTi18} as described earlier.

The fact that we use  Wave, the sender's secret key is a generalized ($U, U+V$)-code over a finite field $\mathbb{F}_q$ with $q>2$. {Its public key is a parity-check matrix of a code equivalent to the previous one.  To reduce the public key size, we use a permuted Goppa subcode for the receiver's public key. Thus, we include the subcode equivalence problem as one of the security assumptions of our scheme.  In Fig. \ref{fig:Com}, we describe the algorithm $\setup$ which will provide common parameters for our scheme.}

\begin{figure}[ht]
	\begin{center}
		\framebox{\parbox{13cm}{
\begin{minipage}[top]{13cm}

\noindent\textbf{\underline{$\setup$}}

\noindent\textit{Input:} ($1^{\lambda}$)

\noindent\textit{Output:}
\begin{itemize}
\item Parameters of sender's generalized ($U, U+V$)-code: code length $n_s$, dimension $k_U$ of U, dimension $k_V$ of V, dimension $k_s=k_U+k_V$ of the  generalized ($U, U+V$)-code, weight of error vector $\omega$, cardinality $q$ of the finite field $\mathbb{F}_q$.
\item Parameters of receiver's Goppa code: degree $m$ of extension $\mathbb{F}_{2^m}$ of $\mathbb{F}_2$, length $n_r$ of the Goppa code, degree  $t$ of Goppa polynomial $g_r$, dimension $\tilde{k}$ of Goppa subcode.

\item A cryptographic hash functions  $\mathcal{H}_1:\{0,1\}^{*} \longrightarrow \{0,1\}^{\tilde{k}}$ 
\item A cryptographic hash functions $\mathcal{H}_0:\{0,1\}^{*} \longrightarrow \{0,1\}^{\ell}$ where $\ell$ is  the bit length  of the symmetric encryption key.
\item A hash function $\mathcal{H}_2:\{0,1\}^{*} \longrightarrow \{0,1,2\}^{r_s}$ where $r_s=n_s-k_s$.
\item A cryptographic hash function $\mathcal{H}_3:\{0,1\}^{*} \longrightarrow \{0,1\}^{\tilde{k}+\ell}$ 
\item An encoding function $\phi:\mathbb{F}_2^{\kappa}\longrightarrow \mathcal{W}_{2,n_r, t}$ where $\kappa$ is a well chosen parameters such that  $\kappa = \left\lfloor \binom{n_r}{t}\right\rfloor$ and $\mathcal{W}_{2,n_r, t}$ is the set of binary vectors of length $n_r$ and Hamming weight $t$.
\end{itemize}
\end{minipage}
}}  
		\end{center}
		\caption{Description of the $\setup$ algorithm for common parameters.}
		\label{fig:Com}
	\end{figure}

We give  key generation algorithms in Figure \ref{fig:Keygen}, where we denote the sender key generation algorithm by $\keygen_s$ and that of the receiver by $\keygen_r$. The receiver algorithm $\keygen_r$ returns as signcryption public key a generator matrix $\mathbf{G}_{\pk,r}\in \mathbb{F}_2^{\tilde{k}\times n_r}$ of a Goppa subcode equivalent. It returns as signcryption secret key the tuple ($g_r, \Gamma_r, \mathbf{S}_r^{-1}, \mathbf{P}_r$), where $\Gamma_r$ and $g_r$ are, respectively, the support and the polynomial of a Goppa code.  $\mathbf{S}_r\in \mathbb{F}_2^{\tilde{k}\times k_r}$ is a full rank matrix and $\mathbf{P}_r$ a permutation matrix. The sender key generation algorithm $\keygen_s$ returns as private key three matrices $\mathbf{S}_s\in \mathbb{F}_3^{(n_s-k_s)\times (n_s-k_s)}$, $\mathbf{H}_{\sk,s}\in \mathbb{F}_3^{(n_s-k_s)\times n_s}$ and $\mathbf{P}_s\in \mathbb{F}_2^{n_s\times n_s}$, where $\mathbf{S}_s\in \mathbb{F}_3^{(n_s-k_s)\times (n_s-k_s)}$ is an invertible matrix, $\mathbf{H}_{\sk,s}\in \mathbb{F}_3^{(n_s-k_s)\times n_s}$ a parity-check matrix of a random generalized ($U, U+V$)-code and $\mathbf{P}\in \mathbb{F}_2^{n_s\times n_s}$ a permutation matrix. The sender public key is a parity-check matrix $\mathbf{H}_{\pk,s}\in \mathbb{F}_3^{(n_s-k_s)\times n_s}$ of a generalized ($U, U+V$) equivalent code given by $\mathbf{H}_{\pk,s}=\mathbf{S}_s\mathbf{H}_{\sk,s}\mathbf{P}_s$.

\begin{figure}[ht]
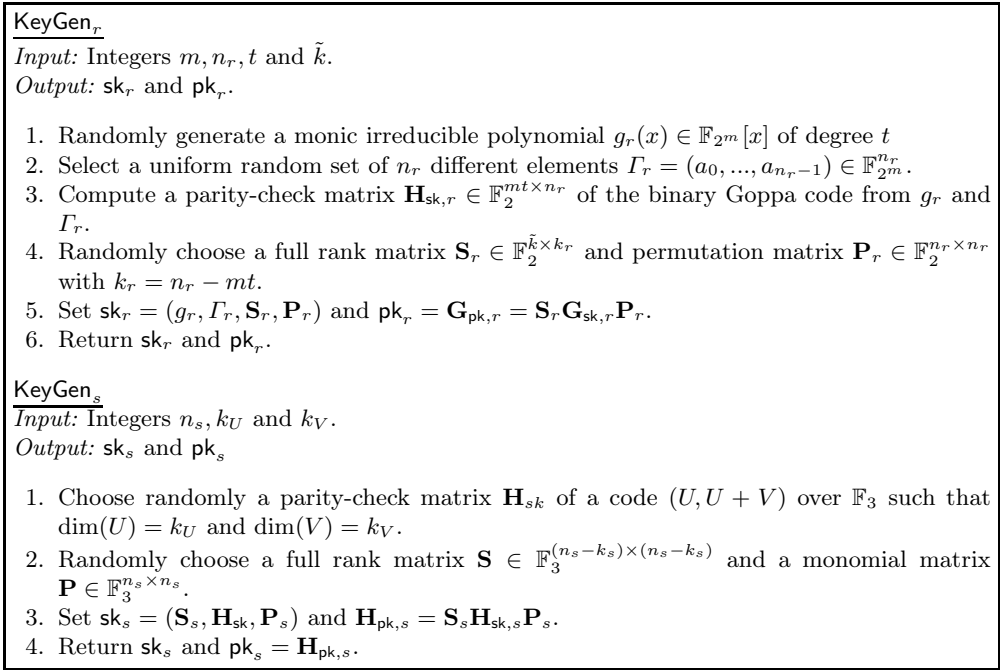

	\begin{center}
		\framebox{\parbox{13cm}{
\begin{minipage}[top]{13cm}

\noindent\textbf{\underline{$\keygen_r$}}

\noindent\textit{Input:} Integers $m, n_r, t$ and $\tilde{k}$.

\noindent\textit{Output:} $\sk_r$ and $\pk_r$.
\begin{enumerate}
\item Randomly generate a monic irreducible polynomial $g_r(x)\in \mathbb{F}_{2^m}[x]$ of degree $t$
\item Select a uniform random set of $n_r$ different elements $\Gamma_r=(a_0,...,a_{n_r-1})\in \mathbb{F}_{2^m}^{n_r}$.
\item Compute a parity-check matrix $\mathbf{H}_{\sk, r}\in \mathbb{F}_2^{mt\times n_r}$ of the binary Goppa code from $g_r$ and $\Gamma_r$.
\item Randomly choose a full rank matrix $\mathbf{S}_r\in \mathbb{F}_2^{\tilde{k}\times k_r}$ and permutation matrix $\mathbf{P}_r\in \mathbb{F}_2^{n_r\times n_r}$ with $k_r=n_r-mt$.
\item Set $\sk_r=(g_r, \Gamma_r, \mathbf{S}_r, \mathbf{P}_r)$ and  $\pk_r=\mathbf{G}_{\pk,r}=\mathbf{S}_r\mathbf{G}_{\sk, r}\mathbf{P}_r$.
\item Return $\sk_r$ and $\pk_r$.
\end{enumerate}

\noindent\textbf{\underline{$\keygen_s$}}

\noindent\textit{Input:} Integers $n_s, k_U$ and  $k_V $.

\noindent\textit{Output:} $\sk_s$ and $\pk_s$
\begin{enumerate}
\item Choose randomly a parity-check matrix $\mathbf{H}_{sk}$ of a code ($U, U+V$) over $\mathbb{F}_3$ such that dim$(U)=k_U$ and dim$(V)=k_V$.
\item Randomly choose a full rank matrix $\mathbf{S}\in \mathbb{F}_3^{(n_s-k_s)\times (n_s-k_s)}$ and a monomial matrix $\mathbf{P}\in \mathbb{F}_3^{n_s\times n_s}$.
\item Set $\sk_s=(\mathbf{S}_s, \mathbf{H}_{\sk},\mathbf{P}_s)$ and $\mathbf{H}_{\pk, s}=\mathbf{S}_s \mathbf{H}_{\sk,s} \mathbf{P}_s$.
\item Return $\sk_s$ and $\pk_s=\mathbf{H}_{\pk, s}$.
\end{enumerate}
\end{minipage}
}}		
\end{center}
\caption{Description of the key generation algorithms.}
\label{fig:Keygen}
	\end{figure}

In Figure \ref{fig:Code-based-SCTKEM2}, we give the design of the symmetric key generation algorithm $\sym$ of our scheme. The algorithm $\sym$ takes as input the bit length $\ell$ of the symmetric encryption key. It outputs an internal state information $\varpi$ and the session key $K$, where $\varpi$ is randomly chosen from $\mathbb{F}_2^{\ell}$, and $K$ is computed by using the hash function $\mathcal{H}_0$.
\begin{figure}[ht]
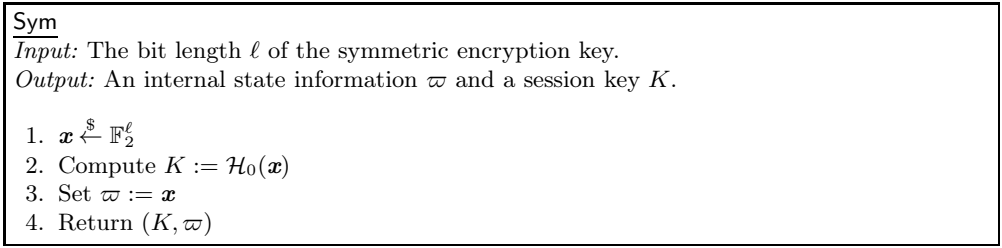

	\begin{center}
		\framebox{\parbox{13cm}{
\begin{minipage}[right]{13cm}

\noindent\textbf{\underline{$\sym$}}

\noindent\textit{Input:} The bit length $\ell$ of the symmetric encryption key. 
 
\noindent\textit{Output:} An internal state information $\varpi$ and a session key $K$.
\begin{enumerate}
\item $\textit{\textbf{x}}\stackrel{\$}{\leftarrow} \mathbb{F}_2^{\ell}$
\item Compute $K:=\mathcal{H}_0(\textit{\textbf{x}})$
\item Set $\varpi:=\textit{\textbf{x}}$
\item Return $(K, \varpi)$
\end{enumerate}
\end{minipage}
}}  
		\end{center}
		\caption{Description of the $\sym$ algorithm.}
		\label{fig:Code-based-SCTKEM2}
	\end{figure}

Figure \ref{fig:Code-based-SCTKEM3} provides a description of the encapsulation  and decapsulation algorithms of our signcryption tag-$\kem$ scheme. We denote the encapsulation algorithm by $\encap$ and the decapsulation by $\decap$. In the encapsulation algorithm, the sender first performs a particular Wave signature on the message $\textit{\textbf{m}}=\tau \Vert \varpi$, where $\varpi$ corresponds to an internal state information and $\tau$ is the input tag. The signature in the Wave scheme  comprises two parts: an error vector $\textit{\textbf{e}}\in \mathbb{F}_3^{n_s}$ and a random binary vector $\textit{\textbf{y}}$. In our scheme, $\textit{\textbf{z}}$ is the hash of a random coin $\textit{\textbf{y}}\in \mathbb{F}_2^{\kappa}$. The sender then performs an encryption of $\textit{\textbf{m}}'=\mathcal{H}_1(\tau)\Vert\varpi$. The encryption that we use in our scheme is  the $\indcca$ secure McEliece encryption scheme with  the Fujisaki-Okamoto transformation introduced by Cayrel et al. \cite{CaHoPe12}. During the encryption, the sender adaptively uses  the random binary vector $\textit{\textbf{y}}$ as a random coin. The resulting ciphertext is denoted by $\textit{\textbf{c}}$.  The output is given by $E=(\textit{\textbf{e}}, \textit{\textbf{c}})$.

\begin{figure}[ht]
	\begin{center}
		\framebox{\parbox{13cm}{
\begin{minipage}[right]{13cm}

\noindent\textbf{\underline{$\encap$}}

\noindent\textit{Input:} ($\varpi$, $\tau$) with $\tau\in \mathbb{F}_2^{\ell}$
\\  \noindent\textit{Output:} An encapsulation of the internal state information $\varpi$.
\begin{enumerate} 
\item $\textit{\textbf{y}}\stackrel{\$}{\leftarrow} \mathbb{F}_{2}^{\kappa}$ with $\kappa= \left\lfloor\log_2\binom{n}{t}\right\rfloor$
\item Compute $\textit{\textbf{z}}=\mathcal{H}_1(\textit{\textbf{y}})$
\item Compute $\textit{\textbf{s}}:=\mathcal{H}_2(\tau \Vert \varpi \Vert \textit{\textbf{z}})$ 
\item Compute $\tilde{\textit{\textbf{e}}}:=\decode_{\mathbf{H}_{\sk,s}}(\textit{\textbf{s}}(\mathbf{S}^{-1})^T)$
\item Compute $\textit{\textbf{e}}:=\tilde{\textit{\textbf{e}}}\mathbf{P}$
\item Compute $\tau'=\mathcal{H}_1(\tau)$
\item Compute $\textit{\textbf{r}}:=\mathcal{H}_1(\tau'\Vert \varpi \Vert \textit{\textbf{y}})$ 
\item Compute $\textit{\textbf{c}}_0:=\textit{\textbf{r}}\mathbf{G}_{\pk,r} +\sigma$, where {$\sigma= \phi(\textbf{\textit{y}})$} with $\phi$ an constant weight encoding function.
\item Compute $\textit{\textbf{c}}_1:=\mathcal{H}_3(\sigma)\oplus (\tau '\Vert \varpi)$.
\item Parse $\textit{\textbf{c}}:=(\textit{\textbf{c}}_0\Vert \textit{\textbf{c}}_1)$  
\item Return $E:=({\textit{\textbf{e}}},\textit{\textbf{c}})$  
\end{enumerate}

\noindent\textbf{\underline{$\decap$}}

\noindent\textit{Input:} ($\sk_r$, $\mathbf{H}_{\pk, s}$, $E$, $\tau$)

\noindent\textit{Output:} Session key $K$
\begin{enumerate}
\item Parse $E$ as $ (\textit{\textbf{e}}, \textit{\textbf{c}})$.
\item Compute $( \textit{\textbf{x}}, \textit{\textbf{y}}):=\decrypt(\sk_r, \textit{\textbf{c}})$
\item Parse $\textit{\textbf{x}}$ as ($\tilde{\tau}\Vert \tilde{\varpi}$)
\item If $\textit{\textbf{e}}\mathbf{H}_{\pk,s}^T \neq \mathcal{H}_2(\tau\Vert \tilde{\varpi}  \Vert \mathcal{H}_1(\textit{\textbf{y}}))$ or $\tilde{\tau}\neq \mathcal{H}_1(\tau)$:
\item \hspace{1cm} Return $\perp$
\item Compute $K:=\mathcal{H}_0(\varpi)$
\item Return $K$.

\end{enumerate}
\end{minipage}
}}  
		\end{center}
		\caption{Description of the $\encap$ and $\decap$ algorithms.}
		\label{fig:Code-based-SCTKEM3}
	\end{figure}

In the decapsulation algorithm $\decap$, the receiver first performs recovery of the internal state information $\varpi$ by using the algorithm $\decrypt$ and the second part of the signature of $\textit{\textbf{m}}$. Then it verifies the signature and computes the session $K$ by using $\varpi$.

The algorithm $\decrypt$ that we use in the decapsulation algorithm of our scheme is described in Figure \ref{fig:Decrypt}. It is similar to that described in \cite{CaHoPe12} but we introduce some modifications which are:
\begin{enumerate}
    \item[•] we use an encoding function $\phi$
    \item[•] the output is not only the clear message $\textit{\textbf{m}}$, but a pair ($\textit{\textbf{m}}, \textit{\textbf{y}}$) where $\textit{\textbf{y}}$ is the reciprocal image the error vector $\sigma$ by the encoding function $\phi$
\end{enumerate}

\begin{figure}[ht]
	\begin{center}
		\framebox{\parbox{13cm}{
\begin{minipage}[right]{13cm}

\noindent\textbf{\underline{$\decrypt$}}

\noindent\textit{Input:} Secrete $\sk=(g_r, \Gamma_r, \mathbf{S}_r,  \mathbf{P}_r)$ the receiver and a ciphertext $\textit{\textbf{c}}$ . 
 
\noindent\textit{Output:} The pair ($\textit{\textbf{x}}, \textit{\textbf{y}}$), where $\textit{\textbf{x}}=\tau'\Vert \varpi$.
\begin{enumerate}
\item Parse $\textit{\textbf{c}}$ into ($\textit{\textbf{c}}_0,\textit{\textbf{c}}_1$) 
\item Compute $\sigma:=\gamma_{\text{Goppa}}^{\text{McE}}(\sk, {\textit{\textbf{c}}}_0)$, where $\gamma_{\text{Goppa}}^{\text{McE}}$ is a decoding algorithm for Goppa code.
\item $\textit{\textbf{y}}=\phi^{-1}(\sigma)$
\item Compute ${\textit{\textbf{x}}}= \textit{\textbf{c}}_1 \oplus\mathcal{H}_3(\sigma)$
\item Compute $\tilde{\textit{\textbf{r}}}=\mathcal{H}_1(\textit{\textbf{x}}\Vert \textit{\textbf{y}})$
\item If $\tilde{\textit{\textbf{r}}}\mathbf{G}_{\pk,r} \oplus \sigma \neq \textit{\textbf{c}}_1:$
\item \hspace{1cm} Return $\perp$
\item  Return $( \textit{\textbf{x}},  {\textit{\textbf{y}}})$
\end{enumerate}
\end{minipage}
}}  
		\end{center}
		\caption{Description of the $\sym$ algorithm.}
		\label{fig:Decrypt}
	\end{figure}

\subsubsection*{Completeness of our signcryption tag-$\kem$}

Let $\tau$ be a tag, ($\sk_s$, $\pk_s$) (resp. $\sk_r$ and $\pk_r$) be sender's  (resp. receiver's) key pair generated by the algorithm $\keygen$ with input $1^{\lambda}$. Let ($K$, $\varpi$):=$\sym$($\sk_s$, $\pk_r$) be a pair of a session key and an internal state information. Let $E:=$($\textit{\textbf{e}}, \textit{\textbf{c}}$) be an encapsulation of the internal state information $\varpi$.  Assuming that the encapsulation and decapsulation are performed by an honest user, we have:
\begin{itemize}
    \item The receiver can recover the pair ($\tau'\Vert \varpi,\textit{\textbf{y}}) $  from $\textit{\textbf{c}}$ and verify successfully  that $$\textit{\textbf{e}}\mathbf{H}_{\pk, s}^T= \mathcal{H}_2(\tau\Vert \varpi \vert \textit{\textbf{y}})\ \ and \ \ \tau '=\mathcal{H}_1(\tau)$$
    Otherwise, the receiver performs a successful signature verification of message $\textit{\textbf{m}}:=\tau\Vert \varpi$ signed by an honest user using the dual version of mCFS signature.
    \item Therefore it can compute the session key $K:=\mathcal{H}_0(\varpi)$.
\end{itemize}

\subsection{Code-based hybrid signcryption}\label{HtagKEM}

Here we use the  signcryption tag-$\kem$ described in Section \ref{tagKEM} for designing a code-based hybrid signcryption.  For the data encapsulation, we propose the use of a regular OT-secure symmetric encryption scheme. We denote the symmetric encryption algorithm being used by $\symencrypt$ and the symmetric decryption algorithm by $\symdecrypt$.

Figure \ref{fig:Code-based Hybrid signcryption from SCTKEM and DEM} gives the design of our code-based hybrid signcryption tag-$\kem$+$\dem$. In this design, algorithms $\setup$, $\keygen_s$ and  $\keygen_r$ are the same as those of our signcrytion tag-$\kem$. Algorithms $\sym$ and $\encap$ are those of our signcryption tag-$\kem$ in Section \ref{tagKEM}. 

\begin{figure}[ht]
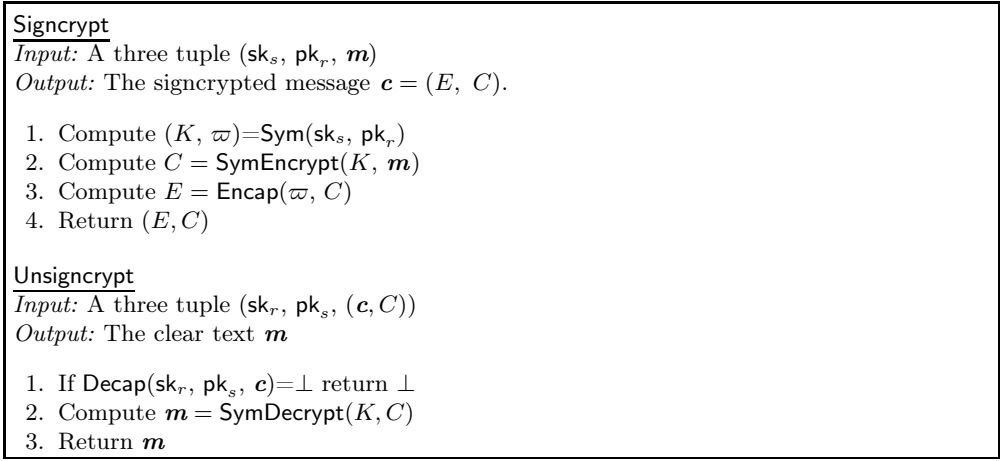

	\begin{center}
		\framebox{\parbox{13cm}{
\vspace{0.2em}

\begin{minipage}[top]{13cm}

\noindent \underline{$\signcrypt$}

\noindent\textit{Input:} A three tuple ($\sk_s$, $\pk_r$, $\textit{\textbf{m}}$)
\\ \textit{Output:} The signcrypted message $\textit{\textbf{c}}=(E,\ C)$. 
\begin{enumerate}
\item Compute ($K$, $\varpi$)=$\sym$($\sk_s$, $\pk_r)$
\item Compute $C=\symencrypt$($K$, \textit{\textbf{m}})
\item Compute $E=\encap$($\varpi$, $C$)
\item Return $(E, C)$
\end{enumerate}

\vspace{0.2em}
\noindent{\underline{$\unsigncrypt$}}

\noindent\textit{Input:} A three tuple ($\sk_r$, $\pk_s$, $(\textit{\textbf{c}}, C)$)
\\ \textit{Output:} The clear text $\textit{\textbf{m}}$ 
\begin{enumerate}

\item If $\decap$($\sk_r$, $\pk_s$, $\textit{\textbf{c}}$)=$\perp$ return $\perp$
\item Compute $\textit{\textbf{m}}=\symdecrypt$($K, C$)
\item Return $\textit{\textbf{m}}$
\end{enumerate}
\end{minipage}
}}  
		\end{center}
		\caption{Code-based hybrid signcryption from $sctkem$ and $\dem$.}
		\label{fig:Code-based Hybrid signcryption from SCTKEM and DEM}
	\end{figure}

\section{Security analysis}\label{Sec:Analysis}

Before discussing the security of our hybrid scheme, let us consider the following assumptions for our security analysis:

\vspace{0.5em}

\noindent{\textit{Assumption 1}}: The advantage of probabilistic polynomial-time algorithm $\mathcal{A}$ to solve the decoding random linear codes problem is negligible with respect to the length $n$ and dimension $k$ of the code.

\vspace{0.5em}

\noindent{\textit{Assumption 2}}:
 The advantage of probabilistic polynomial-time algorithm $\mathcal{A}$ to solve the ($U, U+V$) distinguishing problem is negligible with respect to the length $n$ and dimension $k$ of the code.

\vspace{0.5em}

\noindent{\textit{Assumption 3}}:
The advantage of probabilistic polynomial-time algorithm $\mathcal{A}$ to solve the subcode equivalence problem is negligible with respect to the length $n$ and dimension $k$ of the code.

\vspace{0.5em}
\noindent{\textit{Assumption 4}}:
The advantage of probabilistic polynomial-time algorithm $\mathcal{A}$ to solve the decoding one out of many (DOOM) problem is negligible with respect to the length $n$ and dimension $k$ of the code.

\noindent{{\textit{Assumption 5}}:
The advantage of probabilistic polynomial-time algorithm $\mathcal{A}$ to solve the Goppa code distinguishing problem is negligible with respect to the length $n$ and dimension $k$ of the code.}

\subsection{Information-set decoding algorithm}

In code-based cryptography, the best-known non-structural attacks rely on information-set decoding. The information-set decoding algorithm was introduced by Prange \cite{Prange62} for decoding cyclic codes. After the publication of Prange's work, there have been several works studying to invert code-based encryption schemes based on information-set decoding (see \cite{TeamMcEClassic} Section 4.1).

For a given linear code of length $n$ and dimension $k$, the main idea behind the information-set decoding algorithm is to find a set of $k$ coordinates of a garbled vector that are error-free and such that the restriction of the code’s generator matrix to these positions is invertible. Then, the original message can be computed by multiplying the encrypted vector by the inverse of the submatrix.

Thus, those $k$ bits determine the codeword uniquely, and hence the set is called an information set. It is sometimes difficult to draw the exact resistance to this type of attack. However, they are always lower-bounded by the ratio  of information sets without errors to total possible information sets, i.e.,
\begin{equation}
    R_{\text{ISD}}=\frac{\binom{n-\omega}{k}}{\binom{n}{k}},
\end{equation}
where $\omega$ is the Hamming weight of the error vector. Therefore, well-chosen parameters can avoid these non-structural attacks. In our scheme, we use the parameters of the Wave signature \cite{DeSeTi18} for the sender and those of  Classic McEliece \cite{TeamMcEClassic} for the receiver in the underlying encryption scheme.

\subsection{Key recovery attack}

In code-based cryptography, usually, the first step in the key recovering attack is to perform a distinguishing attack on the public code  in order to identify the family of the underlying code. Once successful, the attacker can then perform any well-known attack against this family of underlying codes to recover the secret key. When the underlying code is a Goppa code, the main {distinguishing} attack technique consists of evaluating the square code or the square of the trace code of the corresponding public code \cite{FaGaOtPeTi13, MaPe12, MoTi22}. { Note that this technique usually works for a Goppa code with a high rate. Compared to many other code-based encryption schemes, in which the public code is equivalent to an alternant or a Goppa code, in this work the public code is a permuted Goppa subcode. Thus, in addition to the indistinguishability of Goppa codes, the subcode equivalence problem becomes one of our security assumptions. Moreover, to the best of our knowledge, there is no attack reported in the literature on distinguishing a code equivalent to a Goppa subcode. Therefore, by using the subcode equivalence problem as a security assumption, we can keep our scheme out of the purview of the distinguishing attack even though the underlying code is a Goppa code.}

{Throughout the rest of our analysis, we assume that the attacker knows that the family of the underlying code is a Goppa code. In our case, the key recovery attack is at two different levels: the first one is on the sender side, and the second one is on the receiver side.}

On the receiver side, the key recovery attack consists of the recovery of the Goppa polynomial $g_r$ and the support $\gamma_r=(\alpha_0, ...,\alpha_{n-1})$ from the public matrix. Therefore, the natural way for this is to perform a brute-force  attack:  one can determine the sequence $(\alpha_0,..., \alpha_{n-1})$ from $g_r$ and the set $\lbrace \alpha_0,..., \alpha_{n-1}\rbrace$, or alternatively determine $g_r$ from $(\alpha_0,...,\alpha_{n-1})$. A good choice of parameters can avoid this attack for the irreducible Goppa code the number of choices of $g_r$ is given by
$$\frac{1}{t}\sum \limits_{d|t}\mu(d)q^{\frac{t}{d}}.$$

By using the parameters of  Classic McEliece, we can see that the complexity for performing a brute-force attack to find Goppa polynomial is more than  $2^{800}$ for the parameters proposed in \cite{TeamMcEClassic}.

{It is also important to note that if the adversary has the knowledge of the underlying Goppa code $\mathcal{C}_{\sk}$, performing the key recovery attack implies solving a computational instance of a subcode equivalence problem. Indeed, this corresponds to finding the permutation $\sigma$ such that $\sigma(\mathcal{C}_{\pk})$ is a subcode of $\mathcal{C}_{\sk}$. We can see that finding the permutation $\sigma$ is equivalent to solving the following system:}

\begin{equation}\label{2}
    \mathbf{G}_{\pk, r}\mathbf{X}_{\sigma}\mathbf{H}_{\sk, r}=\mathbf{0}
\end{equation}
{where $\mathbf{H}_{\sk, r}$ is a parity-check matrix of the underlying Goppa code $\mathcal{C}_{\sk, r}$, $\mathbf{G}_{\sk, r}$ is the generator matrix of the public code $\mathcal{C}_{\pk}$ and $\mathbf{X}_{\sigma}=(x_{i,j})$ is the matrix of the unknown permutation $\sigma$. Note that solving (\ref{2}) is equivalent to solving a variant of permuted kernel problem  \cite{LaPa12}. A natural way to solve (\ref{2}) is to use the brute force attack and such an attack is of order $\mathcal{O}(n!)$. However, the adversary could use  Georgiades' technique \cite{Geo92} where its complexity is given in our case by}
\begin{equation}\label{3}
    \mathcal{O}\left( \frac{n!}{\tilde{k}!}\right).
\end{equation}

{Recently Paiva and Terada  introduced in \cite{PaTe21} a new technique for solving (\ref{2}). The workfactor of their attack applied to our scheme is given by:}
\begin{equation}\label{4}
 \wf_{\attack_{\pate}}=\mathcal{O}\left( 2^{(n-mt-\tilde{k}n^{-1/5})(\lceil \log(n)\rceil-1)-0.91n+\frac{\log n}{2}} \right) 
\end{equation}

{From (\ref{3}) and (\ref{4}), we can see that a well-chosen set of parameters can avoid the attack of Georgiades as well as that of Paiva and Terada.}
 

In the case of the sender, the key recovery attack consists of first solving the ($U, U+V$) distinguishing problem for finite fields of cardinality $q=3$. Therefore under Assumption 3 and with a well-chosen set of parameters, this attack would fail.

\subsection{IND-CCA and SUF-CMA security}

In code-based cryptography, the main approach to a chosen-ciphertext attack against the McEliece encryption scheme consists of adding two errors to the received word. If the decryption succeeds, it means that the error vector in the resulting word has the same weight as the previous one. In our signcryption tag-$\kem$ scheme, this implies either recovering the session key  $K$ or distinguishing encapsulation of two different session keys  from $(\textit{\textbf{e}}, \textit{\textbf{c}}, \tau)$. We see that the recovery of the session key $K$ corresponds to the recovery of plaintext in  a $\indcca$ secure version of {McEliece}'s cryptosystem (see \cite{CaHoPe12} Subsection 3.2).  We now have the following theorem:

{\begin{theorem}\label{Secuproof1}
Under Assumptions 1, 3, and 5, the signcryption {tag}-$\kem$ scheme described in Subsection \ref{tagKEM} is $\indcca$ secure.
\end{theorem}}

\begin{proof}
Let $\mathcal{A}_{\cca}$ be a PPT adversary against the signcryption tag-$\kem$ scheme described in Subsection \ref{tagKEM} in the signcryption tag-$\kem$ $\indcca$ game. Let us denote its advantage by $\epsilon_{\cca,\sctkem}$. For proving Theorem \ref{Secuproof1} we need to bound $\epsilon_{\cca,\sctkem}$.

\vspace{0.5em}
\noindent{\textbf{Game 0:} This game is the normal signcryption tag-$\kem$ $\indcca$ game. Let us denote by $X_0$ the event that the adversary wins Game 0 and $\pr (X_0)$ the probability that it happens. Then we have
$$
\pr (X_0)=\epsilon_{\cca,\sctkem}
$$}

\noindent{{\textbf{Game 1:}} This game corresponds to the simulation of  the hash function oracle. Indeed it is {the same as} Game 0 except that adversary can have access to the hash function oracle: It looks for some pair $(\tau^{*}, \textit{\textbf{y}}^{*})\in \mathbb{F}_2^{\lambda}\times \mathbb{F}_2^{\kappa}$ such that $\textit{\textbf{e}}\mathbf{H}_s^T=\mathcal{H}_2(\tau^{*}\Vert \varpi \Vert  \mathcal{H}_1(\textit{\textbf{y}}^{*}))$. Then, it tries to continue by computing  $\textit{\textbf{c}}'$. We can see that it could succeed at least when the following collisions happen: 
$$\mathcal{H}_1(\tau^{*})=\mathcal{H}_1(\tau)\ \ \text{and} \ \ \mathcal{H}_1(\tau^{*}\Vert \varpi\Vert \mathcal{H}_2(\textit{\textbf{y}}^{*}))=\mathcal{H}_2(\tau \Vert \varpi \Vert \mathcal{H}_1(\textit{\textbf{y}}))$$
Therefore, if $q_h$ is the number of queries allowed and $X_1$ the event that $\mathcal{A}_{\cca}$ wins game $X_1$, then we have:
$$
|\pr (X_1)-\pr (X_0)| \leq \frac{q_h}{\binom{n}{t}}
$$
\noindent{\textbf{Game 2:}} This game is the same as Game 1 except that the error vector  $\textit{\textbf{e}}$ in the encapsulation output is generated randomly. We can see that the best to proceed is to split $\textit{\textbf{c}}$ as $(\textit{\textbf{c}}_0\Vert \textit{\textbf{c}}_1)$ and then try to invert either $\textit{\textbf{c}}_0$ for recovering the error $\sigma$ or $\textit{\textbf{c}}_1$ for recovering directly the internal state $\varpi_b$. That means that the adversary is able either to solve the syndrome decoding problem or to invert a one-time pad function. Therefore we have:
$$
|\pr (X_1)-\pr (X_2)|\leq \epsilon_{\sd}+\nu(\ell)
$$
where $\epsilon_{\sd}$ is the advantage of an adversary against the syndrome decoding problem, $\nu$ is a negligible function, and $\ell$ is the bit length of the symmetric encryption. }

\noindent{{\textbf{Game 3:}} This game is the same as Game 2. However, the change is in the key generation algorithm. Indeed, a random code is chosen as the underlying code instead of Goppa. We can see that this change is indistinguishable. In fact, distinguishing this change corresponds to solving in part the Goppa code distinguishing problem. Thus, we have
$$
|\pr (X_3)-\pr (X_2)| \leq \epsilon_{GCD}(\lambda)
$$
where $\epsilon_{GCD}(\lambda)$ is the advantage of a PPT adversary in the Goppa code distinguishing problem and   $\lambda$ the security parameter. If there is a PPT adversary $\mathcal{A}$ capable of distinguishing this change, we can use it to construct an adversary $\mathcal{A}_{GCD}$ to solve the Goppa code distinguishing problem as follows: 
\begin{enumerate}
    \item Once receiving an instance $\mathbf{G}\in \mathbb{F}_2^{k\times n}$ of a generator matrix of a code $\mathcal{C}$ in Goppa code distinguishing problem, $\mathcal{A}_{GCD}$ extracts a generator matrix $\mathbf{G}'$ of a subcode $\mathcal{C}'$ of $\mathcal{C}$ and forward it to $\mathcal{A}$. 
    \item $\mathcal{A}$ will reply by $1$ if the change has happened, i.e.,  the underlying code is not a Goppa code. It will reply by 0 otherwise.
    \item If $\mathcal{A}_{GCD}$ receives $1$ from $\mathcal{A}$, it means that $\mathcal{C}$ is not a Goppa code and $\mathcal{A}_{GCD}$ outputs $0$, otherwise it returns $1$, i.e, $\mathcal{C}$ is a Goppa code.
\end{enumerate}}

\noindent{\textbf{Game 4:} This game  is the same as Game 3 except that the public key is a random matrix instead of a generator matrix of a permuted subcode. We can see that this change is indistinguishable according to the subcode equivalence assumption. Thus we have:
$$
|\pr (X_4)-\pr (X_3)| \leq \epsilon_{ES}(\lambda)
$$
where $\epsilon_{ES}(\lambda)$ is the advantage of a PPT adversary in the subcode equivalence problem and $\lambda$ is the security parameter. Moreover, we can show that if an adversary $\mathcal{A}_{\cca}$ wins this game, we can use  it to construct an adversary $\mathcal{A}_{\text{McE}}$  for attacking the underlying  McEliece scheme in the public key encryption $\indcca$ game (called $\pkegame$ in Appendix \ref{pkegame}). For more details on the underlying McEliece encryption scheme and its $\indcca$ security proof, the reader is referred to Appendix \ref{McElie}. We now proceed as follows:
\begin{enumerate}
    \item[•] Given the receiver public key $\pk$ which corresponds to a receiver public key signcryption tag-$\kem$, $\mathcal{A}_{\text{McE}}$ does the following:
    \begin{itemize}
        \item[$\star$] chooses randomly $(\varpi_0, \varpi_1)\stackrel{\$}{\leftarrow}\mathbb{F}_2^{\ell}$
        \item[$\star$] chooses randomly $\delta \stackrel{\$}{\leftarrow} \{0,1\}$
        \item[$\star$] sends the public key $\pk$ and $\varpi_{\delta}$ to $\mathcal{A}_{\cca}$
    \end{itemize}
    \item[•] Given a tag $\tau$ from $\mathcal{A}_{\cca}$,  $\mathcal{A}_{\text{McE}}$:
    \begin{itemize}
        \item[$\star$]  sends the pair ($\mathcal{H}_1(\tau)\Vert \varpi_0$,$\mathcal{H}_1(\tau)\Vert \varpi_1$) to the encryption oracle of $\pkegame$
        \item[$\star$] forwards $\textit{\textbf{c}}$ received from the encryption oracle to $\mathcal{A}_{\cca}$
    \end{itemize}
    \item[•] For every decryption query ($\textit{\textbf{c}}_i$, $\tau_i$) from $\mathcal{A}_{\cca}$:
    \begin{itemize}
        \item[$\star$]  if $\textit{\textbf{c}}_i=\textit{\textbf{c}}$, $\mathcal{A}_{\text{McE}}$ return $\perp$ to $\mathcal{A}_{\cca}$, otherwise it sends $\textit{\textbf{c}}_i$ to the decryption oracle of $\pkegame$. 
        \item[$\star$] Receiving $\tau_i'\Vert \varpi_i$ from the decryption oracle:
        \begin{itemize}
        \item[$\triangleright$]  if $\tau_i'\neq \mathcal{H}_1(\tau _i)$, $\mathcal{A}_{\text{McE}}$ returns $\perp$ to $\mathcal{A}_{\cca}$, otherwise, it returns $\varpi_i$ to $\mathcal{A}_{\cca}$
    \end{itemize}
    \end{itemize}
    \item[•] When $\mathcal{A}_{\cca}$ outputs $\tilde{\delta}=\delta$, $\mathcal{A}_{\text{McE}}$ returns 1, otherwise, it returns 0.
\end{enumerate}}

{Let $\epsilon_{\text{PKE}}$ be the advantage of $\mathcal{A}_{\text{McE}}$ in the $\pkegame$. Note that the target ciphertext $\textit{\textbf{c}}$ can be uniquely decrypted to $\mathcal{H}_1(\tau))\Vert \varpi_{\delta}$. Therefore any $(\textit{\textbf{c}}, \tau')$ other than $(\textit{\textbf{c}}, \tau)$ cannot be a valid signcryption ciphertext unless collusion of $\mathcal{H}_1$ takes place, i.e., $\mathcal{H}_1(\tau_i)=\mathcal{H}_1(\tau)$. The correct answer to any decryption query with $\textit{\textbf{c}}_i=\textit{\textbf{c}}$ is $\perp$. Decryption queries from $\mathcal{A}_{\cca}$ are correctly answered since $\textit{\textbf{c}}_i$ is decrypted by the decryption oracle of $\pkegame$.}

{When $\mathcal{A}_{\cca}$ outputs $\tilde{\delta}$, it means that $\varpi_{\delta}$ is embedded in $\textit{\textbf{c}}_i$ otherwise $\varpi_{1-\delta}$ is embedded. It means that the adversary $\mathcal{A}_{\text{McE}}$ wins game $\pkegame$ with the same probability as $\mathcal{A}_{\cca}$ wins Game 4 when collision of $\mathcal{H}_1$  has  happened. Let $\tilde{X}$ be the event collision of $\mathcal{H}_1$ has  happened and $\tilde{X}_4$ the event $\mathcal{A}_{\text{McE}}$ wins the $\pkegame$. Let us denote by $\epsilon_{pke}$ the probability of the event $\tilde{X}_4$ and $\epsilon_{col}$ that of  $\tilde{X}$. Therefore we have:
$$
\pr (X_4|\tilde{X})=\pr (\tilde{X}_4)\Longrightarrow \pr (X_4)\leq \pr (\tilde{X}_4)+\pr (\tilde{X})
$$
By putting it  all together, we conclude our proof.}
\end{proof}
\begin{theorem}\label{Secuproof11}
Under Assumptions 2 and 4, the signcryption tag-$\kem$ scheme described in Subsection \ref{tagKEM} is $\sufcma$ secure.
\end{theorem}

\begin{proof}
Let $\mathcal{F}_{\cma}$ be an adversary against our signcryption tag-$\kem$ in the $\sufcma$ game and $\epsilon_{\cma}$ its advantage. 
For the forgery of our signcryption, adversary $\mathcal{F}_{\cma}$ needs to first find a pair $(\textit{\textbf{e}}, \textit{\textbf{y}})\in \mathcal{W}_{q, n,\omega}\times \mathbb{F}_2^{\tilde{k}}$ such that $\textit{\textbf{e}}\mathbf{H}_{\pk,s}^T=\mathcal{H}_2(\tau\Vert \varpi\Vert \textit{\textbf{y}})$. 
Then, it will try to find $\textit{\textbf{r}}\in \mathbb{F}_2^{\kappa}$ such that $\mathcal{H}_1(\textit{\textbf{r}})=\textit{\textbf{y}}$, i.e., it wins in the target pre-image free game (see Appendix \ref{PreImage}) against the cryptographic hash function $\mathcal{H}_1$. We can see that finding $(\textit{\textbf{e}}, \textit{\textbf{y}})\in \mathcal{W}_{q, n,\omega}\times \mathbb{F}_2^{\tilde{k}}$ such that $\textit{\textbf{e}}\mathbf{H}_{\pk,s}^T=\mathcal{H}_2(\tau\Vert \varpi\Vert \textit{\textbf{y}})$ corresponds to the forgery of the underlying Wave signature scheme. Let $\epsilon_{\text{PreIm}}$ be the advantage of an adversary in the pre-image free game against a cryptographic hash function. Let $\mathcal{A}_{\text{Wave},\cma}$ be an adversary against the Wave signature in the ${\eufcma}$ game and $\epsilon_{Wave, EUF}$ its advantage. Let $X$ be the event that $\mathcal{A}_{\text{Wave},\cma}$ wins. Let $\tilde{X}$ be the event that the adversary is able to find a pre-image $\textit{\textbf{x}}$ of $\textit{\textbf{y}}$ by $\mathcal{H}_1$ such that $\textit{\textbf{x}}\in \mathbb{F}_2^{\kappa}$. We have:

$$\pr (\mathcal{F}_{\cma}\  \mbox{wins})= \pr (X \ \mbox{and} \ \tilde{X})$$ 
$$\hspace{3cm}\leq \pr (X)+\pr (\tilde{X})$$
$$\hspace{3.8cm}\leq \epsilon_{\text{Wave}, \euf}+\dfrac{\epsilon_{\text{PreIm}}}{2^{\kappa}}$$

Note that due to the fact that $\mathcal{H}_1$ is a cryptographic hash function, $\epsilon_{\text{PreIm}}$ is negligible and  that concludes our proof.
\end{proof}

\begin{corollary}
The signcryption tag-$\kem$ described in  Subsection \ref{tagKEM} is secure.
\end{corollary}

The above corollary is a consequence of Theorems  \ref{Secuproof1} and  \ref{Secuproof11}. We then have the following.

{\begin{proposition}\label{Secuproof21}
Under Assumptions 1, 3, and 5, the hybrid signcryption tag-$\kem+\dem$ scheme described in Subsection \ref{HtagKEM} is $\indcca$.
\end{proposition}}
\begin{proof}$\ $

\noindent{Proposition \ref{Secuproof21} is a consequence of Theorem \ref{Genesec}. Indeed, under Assumptions 1, 3, and 5, the underlying signcryption tag-$\kem$ is $\indcca$ secure (see Theorem \ref{Secuproof1}). In addition, the symmetric encryption scheme used is OT-secure. Therefore, a direct application of  Theorem \ref{Genesec} allows us to achieve the proof.}
\end{proof}
{\begin{proposition}\label{Secuproof22}
Under Assumptions 2 and 4, the hybrid signcryption tag-$\kem+\dem$ scheme described in Subsection \ref{HtagKEM} is $\sufcma$ secure.
\end{proposition}}
\begin{proof}$\ $

\noindent{Under Assumptions 2 and 4, the underlying signcryption tag-$\kem$ is $\sufcma$ secure and,  therefore, according to the Theorem \ref{Genesec}, the proposed   hybrid signcryption tag-$\kem+\dem$ is $\sufcma$ secure.}
\end{proof}
\section{Parameter values}\label{Param}

For our scheme, we choose parameters such that $\lambda_0=\lambda+2\log_2(q_{\sig})$ and  $\lambda_{\text{McE}}$ of the underlying Wave signature and McEliece's encryption, respectively,  satisfy $\max(\lambda_0, \lambda_{\text{McE}})\leq \left\lfloor \binom{n_r}{t}\right\rfloor$. According to the sender and receiver keys, the size of our ciphertext is given by 
$$
|E|=|\textit{\textbf{e}}|+|\textit{\textbf{c}}|+|C|=2n_s+n_r+\tilde{k} + 2\ell.
$$

Table~\ref{tab:Param} gives suggested values of the parameters of our scheme. These values have been derived using those of  Wave \cite{BaDeNe21} and Classic McEliece  \cite{TeamMcEClassic} for NIST PQC Level 1 security. According to the values  given in  Table \ref{tab:Param}, the ciphertext size in bits of our scheme is in the order of $|E|=2.9\times 10^{4}$.

\begin{table}[ht]
	\begin{center}
\begin{tabular}{lccccccccc}

   Parameter & $n_s$&$k_U$& $k_V$&$\omega$&$m$& $t$&$n_r$& $\tilde{k}$& $\ell$\\
   Value & 8492  & 3558  & 2047  & 7980 & 12 & 64  & 3488  &  1815  & 512 \\ 
\end{tabular}
			\end{center}
			\caption{Parameter values of the proposed scheme.}
			\label{tab:Param}
\end{table}
\vspace{-0.5cm}
Table  \ref{fig:keysize} provides key sizes of our scheme in terms of relevant parameters. Then in Table \ref{fig:comparison}  we give a numerical comparison of key and ciphertext sizes of our scheme with some existing lattice-based hybrid signcryption schemes. The rationale behind comparing our scheme against lattice-based schemes is that no code-based hybrid signcryption scheme exists in the literature and the underlying hard problems in both codes- and lattice-based schemes are considered quantum-safe. For the lattice-based schemes in our comparison, the parameters, including plaintext size of 512 bits, are from  \cite[Table 2]{SaSh18}. We can see that for post-quantum security level 1 the proposed scheme 
has the smallest  key and ciphertext sizes.

\begin{table}[ht]
	\begin{center}
\begin{tabular}{lcc}
   User &  Public key & Secret key \\ 
   Receiver's key size & $\tilde{k}n_r$ & $m(2n_r+t-\tilde{k}t)+\tilde{k}n_r$ \\
   Sender's key size & $r(n_s-r)\log_2(q)$ &  $(n_s(n_s+r)+r^2)\log_2(q)$ \\
\end{tabular}
	\end{center}
\caption{Key sizes of the proposed scheme.}
	\label{fig:keysize}
\end{table}
\vspace{-0.5cm}

\begin{table}[ht]
	\begin{center}
\begin{tabular}{lccccc}
    {Construction} & \multicolumn{2}{c}{Receiver's key size}& \multicolumn{2}{c}{Sender's key size}&{$\ $Ciph. size$\ $} \\
    & Pub. key & Sec. key&Pub. key & Sec. key& \\
   $SC_{TK}$ \cite{SaSh18,ChMaScMa11} & $8.5\times10^{7}$ & $4.2\times10^8$ & $8.4\times 10^7$ & $4.2\times 10^8$  & $5.5\times 10^{5}$ \\
    $SC_{KEM}$ \cite{SaSh18,ChMaScMa11}& $5.7\times10^7$ & $4.2\times 10^8$ & $8.5\times 10^7$ & $4.2\times 10^8$ & $5.2\times 10^{5}$\\
   $SC_{CHK}$ \cite{SaSh18,NaSh13} & $2.8\times 10^7$ & $4.2\times 10^8$ & $2.8\times 10^7$ & $4.2\times 10^8$ & $4.5\times 10^{6}$\\
Shingo and Junji \cite{SaSh18} & $2.8\times 10^7$ & $4.2\times 10^8$ & $2.8\times 10^7$ & $4.2\times 10^8$ &  $4.0\times 10^{5}$\\
   \hline
   Our scheme & $6.3\times 10^6$ & $5.0\times 10^6$ & $2.6\times 10^7$ & $1.7\times 10^8$ & $2.1\times 10^{4}$ \\
\end{tabular}
	\end{center}
\caption{Size comparison (in bits) of the proposed scheme with the lattice-based schemes of \cite{SaSh18,NaSh13,ChMaScMa11}.}
	\label{fig:comparison}
\end{table}

\section{Conclusion}\label{Conclusion}

In this paper, we have proposed a new 
signcryption tag-$\kem$ based on coding theory. The security of our scheme relies on known hard problems in coding theory.
We have used the proposed signcryption scheme to design a new code-based hybrid signcryption tag-$\kem$+$\dem$.   We have proven that the proposed schemes are $\indcca$ and $\sufcma$ secure against any probabilistic polynomial-time adversary. The proposed scheme has  a smaller ciphertext size compared to the pertinent lattice-based schemes.



\begin{thebibliography}{100}


\bibitem{Durendal18} N. Aragon, O. Blazy, P. Gaborit, A. Hauteville, and G. Zémor, ``Durandal: a rank metric based signature scheme,'' in Annual International
Conference on the Theory and Applications of Cryptographic Techniques. Springer, 2019, pp. 728–758.

\bibitem{TeamMcEClassic} M. R. Albrecht, D. J. Bernstein et al., ``Classic McEliece: conservative code-based cryptography.'' Online Available:  \url{https://classic.mceliece.org/nist/mceliece-20201010.pdf} 


\bibitem{BaStZh07} J. Baek, R. Steinfeld, and Y. Zheng, ``Formal proofs for the security of signcryption,'' Journal of Cryptology, vol. 20, no. 2, pp. 203–235,
2007.

\bibitem{Dags18} G. Banegas, P.~S. Barreto, B.~O. Boidje, P.-L. Cayrel, G.~N. Dione, K.~Gaj, C.~T. Gueye, R. Haeussler, J.~B. Klamti, O. N’diaye et al.,
``Dags: Key encapsulation using dyadic gs codes,'' Journal of Mathematical Cryptology, vol. 12, no. 4, pp. 221–239, 2018.

\bibitem{BaDeNe21} G. Banegas, T. Debris-Alazard, M. Nedeljkovic, and B. Smith, “Wavelet: Code-based postquantum signatures with fast verification on
microcontrollers,” arXiv preprint arXiv:2110.13488, 2021.

\bibitem{BaCo18} E. Barelli and A. Couvreur, ``An efficient structural attack on {NIST} submission DAGS,'' in International Conference on the Theory and
Application of Cryptology and Information Security. Springer, 2018, pp. 93–118

\bibitem{BaLiMcQu10} P. S. Barreto, B. Libert, N. McCullagh, and J.-J. Quisquater, ``Signcryption schemes based on the Diffie–Hellman problem,'' in Practical Signcryption. Springer, 2010, pp. 57–69.

\bibitem{BaLiBMcQu10.1} P.~S. Barreto, B. Libert, N. McCullagh and J. Quisquater, ``Signcryption schemes based on bilinear maps,'' in Practical Signcryption. Springer, 2010, pp. 71–97. 

\bibitem{BarLinMiso2011}  P. S. Barreto, R. Lindner, and R. Misoczki, ``Monoidic codes in cryptography,'' in International Workshop on Post-Quantum
Cryptography. Springer, 2011, pp. 179–199.

\bibitem{BGK17} T. P. Berger, C. T. Gueye, and J. B. Klamti, ``A NP-complete problem in coding theory with application to code based cryptography,''
in International Conference on Codes, Cryptology, and Information Security. Springer, 2017, pp. 230–237.

\bibitem{Berg05} T.~P. Berger and P. Loidreau, ``How to mask the structure of codes for a cryptographic use,'' Designs, Codes and Cryptography, vol. 35,
no. 1, pp. 63–79, 2005.

\bibitem{BCG09} T. P. Berger, P.-L. Cayrel, P. Gaborit, and A. Otmani, ``Reducing key length of the McEliece cryptosystem,'' in International Conference
on Cryptology in Africa. Springer, 2009, pp. 77–97.

\bibitem{BeMcVa78} E. Berlekamp, R. McEliece, and H. Van Tilborg, ``On the inherent intractability of certain coding problems (corresp.),'' IEEE Transactions
on Information Theory, vol. 24, no. 3, pp. 384–386, 1978.

\bibitem{BjDe06} T.~E. Bjørstad and A.~W. Dent, ``Building better signcryption schemes with tag-$\kem$s,'' in International Workshop on Public Key
Cryptography. Springer, 2006, pp. 491–507.

\bibitem{BiMiPeSa20}  J.-F. Biasse, G. Micheli, E. Persichetti, and P. Santini, “Less is more: Code-based signatures without syndromes.” in  International Conference on Cryptology in Africa. Springer, Cham, 2020. p. 45-65.



\bibitem{CaHoPe12} P.-L. Cayrel, G. Hoffmann, and E. Persichetti, ``Efficient implementation of a CCA2-secure variant of McEliece using generalized
Srivastava codes,'' in International Workshop on Public Key Cryptography. Springer, 2012, pp. 138–155.

\bibitem{CaGuNdPe17} P.-L. Cayrel, C. T. Gueye, O. Ndiaye, E. Persichetti et al., “Efficient implementation of hybrid encryption from coding theory,” in
International Conference on Codes, Cryptology, and Information Security. Springer, 2017, pp. 254–264.

\bibitem{CVE10} P.-L. Cayrel, P. Véron, and S. M. E. Y. Alaoui, ``A zero-knowledge identification scheme based on the q-ary syndrome decoding
problem,'' in International Workshop on Selected Areas in Cryptography. Springer, 2010, pp. 171–186.

\bibitem{ChMaScMa11}  D. Chiba, T. Matsuda, J. C. Schuldt, and K. Matsuura, ``Efficient generic constructions of signcryption with insider security in the
multi-user setting,'' in International Conference on Applied Cryptography and Network Security. Springer, 2011, pp. 220–237. 

\bibitem{CFS01} N. T. Courtois, M. Finiasz, and N. Sendrier, ``How to achieve a McEliece-based digital signature scheme,'' in International Conference
on the Theory and Application of Cryptology and Information Security. Springer, 2001, pp. 157–174.

\bibitem{CrSh03} R. Cramer and V. Shoup, ``Design and analysis of practical public-key encryption schemes secure against adaptive chosen ciphertext
attack,'' SIAM Journal on Computing, vol. 33, no. 1, pp. 167–226, 2003.


\bibitem{DeSeTi17} T. Debris-Alazard, N. Sendrier, and J.-P. Tillich, ``The problem with the SURF scheme,'' arXiv preprint arXiv:1706.08065, 2017.

\bibitem{DeSeTi18} T. Debris-Alazard, N. Sendrier, and J.-P. Tillich, ``Wave: A new code-based signature scheme,'' Technical report, Cryptology ePrint Archive: Report 2018/996, 2018, https://eprint.iacr.org/2018/996/20181022:154324

\bibitem{DeSeTi18a} T. Debris-Alazard, N. Sendrier, and J.-P. Tillich, ``Wave: A new family of trapdoor one-way preimage sampleable functions based on
codes,'' in International Conference on the Theory and Application of Cryptology and Information Security. Springer, 2019, pp. 21–51.

\bibitem{Dent04}  A.~W. Dent, ``Hybrid cryptography,''Technical report,  Cryptology ePrint Archive, Report 2004/210, 2004, \url{https://eprint.iacr.org/2004/210}.

\bibitem{Dent05} A.~W. Dent, ``Hybrid signcryption schemes with insider security,'' in Australasian Conference on Information Security and Privacy. Springer,
2005, pp. 253–266. 

\bibitem{Dent051}  A.~W. Dent, ``Hybrid signcryption schemes with outsider security,'' in International Conference on Information Security. Springer, 2005, pp.
203–217. 

\bibitem{DeMa10} A.~W. Dent and J. Malone-Lee, “Signcryption schemes based on the RSA problem,” in Practical Signcryption. Springer, 2010, pp.
99–117.

\bibitem{YuDeZh} A. W. Dent and Y. Zheng, Eds., Practical signcryption, ser. Information Security and Cryptography. Springer, 2010.




\bibitem{FaGaOtPeTi13} J.-C. Faugere, V. Gauthier-Umana, A. Otmani, L. Perret, and J.-P. Tillich, ``A distinguisher for high-rate McEliece cryptosystems,''
IEEE Transactions on Information Theory, vol. 59, no. 10, pp. 6830–6844, 2013.

\bibitem{Fiat86}  A. Fiat and A. Shamir, “How to prove yourself: Practical solutions to identification and signature problems,” in Conference on the
theory and application of cryptographic techniques. Springer, 1986, pp. 186–194.

\bibitem{RaCoSS17} K. Fukushima, P. S. Roy, R. Xu, S. Kiyomoto, K. Morozov, and T. Takagi, “Racoss: Random code-based signature scheme,” Submission
to NIST post-quantum standardization process, 2017.

\bibitem{FuOk99} E. Fujisaki and T. Okamoto, ``Secure integration of asymmetric and symmetric encryption schemes,'' in Annual International Cryptology
Conference. Springer, 1999, pp. 537–554. 

\bibitem{Geo92} J. Georgiades, ``Some remarks on the security of the identification scheme based on
permuted kernels,'' Journal of Cryptology 5(2), 133-137 (Jan 1992)



\bibitem{JaKrPiTe12} A. Jain, S. Krenn, K. Pietrzak, and A. Tentes, ``Commitments and efficient zero-knowledge proofs from learning parity with noise,'' in
International Conference on the Theory and Application of Cryptology and Information Security. Springer, 2012, pp. 663–680.

\bibitem{JoJo02} T. Johansson and F. Jonsson, ``On the complexity of some cryptographic problems based on the general decoding problem,'' IEEE
Transactions on Information Theory, vol. 48, no. 10, pp. 2669–2678, 2002.



\bibitem{LaPa12} R. ~ Lampe, J. ~ Patarin, ``Analysis of some natural variants of the PKP algorithm.''
In: Proceedings of the International Conference on Security and Cryptography - Volume 1: SECRYPT, (ICETE 2012). pp. 209-214. INSTICC, SciTePress (2012).

\bibitem{LeDuRoSuFuKi21} H. ~ Q. Le, D.~H. Duong, P. ~S. Roy, W. Susilo, K. Fukushima, S. Kiyomoto,  ``Lattice-based signcryption with equality test in standard model.'' Computer Standards \& Interfaces p.103515 (2021).

\bibitem{LiMuKhTa} F. Li, F. T. Bin Muhaya, M. K. Khan, and T. Takagi, ``Lattice-based signcryption,'' Concurrency and Computation: Practice and Experience, vol. 25, no. 14, pp. 2112–2122, 2013. 

\bibitem{Lyub09} V. Lyubashevsky, ``Fiat-shamir with aborts: Applications to lattice and factoring-based signatures,'' in International Conference on the
Theory and Application of Cryptology and Information Security. Springer, 2009, pp. 598–616.

\bibitem{LiXiYe20} Z. Li, C. Xing, and S. L. Yeo, ``A new code based signature scheme without trapdoors.'' IACR Cryptol. ePrint Arch., vol. 2020, p.
1250, 2020. 


\bibitem{MaVaRa12}  K.~P. Mathew, S. Vasant, and C.~P. Rangan, ``On provably secure code-based signature and signcryption scheme,'' IACR Cryptology
ePrint Archive, vol. 2012, p. 585, 2012.

\bibitem{MaVaRa13} K. P. Mathew, S. Vasant, and C. P. Rangan, ``Efficient code-based hybrid and deterministic encryptions in the standard model,'' in
International Conference on Information Security and Cryptology. Springer, 2013, pp. 517–535. 

\bibitem{MaPe12}  I. Márquez-Corbella and R. Pellikaan, ``Error-correcting pairs for a public-key
cryptosystem,'' CBC 2012, Code-based Cryptography Workshop, 2012. Available on
\url{http://www.win.tue.nl/~ruudp/paper/59.pdf}.

\bibitem{Mce}  R. J. McEliece, ``A public-key cryptosystem based on algebraic coding theory,'' DSN progress report, pp. 42–44, 1978.

\bibitem{MisBar09} R. Misoczki and P. S. Barreto, ``Compact McEliece keys from Goppa codes,'' in International Workshop on Selected Areas in
Cryptography. Springer, 2009, pp. 376–392. 

\bibitem{MisTiSenBar2013} R. Misoczki, J.-P. Tillich, N. Sendrier, and P. S. Barreto, ``MDPC-McEliece: New McEliece variants from moderate density parity-check
codes,'' in 2013 IEEE international symposium on information theory. IEEE, 2013, pp. 2069–2073. 




\bibitem{MoTi22} R. Mora and J.-P. Tillich, Jean-Pierre, ``On the dimension and structure of the square of the dual of a Goppa code,'' WCC 2022: The Twelfth International Workshop on Coding and Cryptography March 7 - 11, 2022, Rostock (Germany). Available on \url{https://www.wcc2022.uni-rostock.de/storages/uni-rostock/Tagungen/WCC2022/Papers/WCC_2022_paper_68.pdf}


\bibitem{Nied1986} H. Niederreiter, ``Knapsack-type cryptosystems and algebraic coding theory,'' Prob. Control and Inf. Theory, vol. 15, no. 2, pp. 159–166,
1986.

\bibitem{NaSh13} R. Nakano and J. Shikata, “Constructions of signcryption in the multi-user setting from identity-based encryption,” in IMA International
Conference on Cryptography and Coding. Springer, 2013, pp. 324–343.



\bibitem{PaTe21} T. ~ B. ~ Paiva, R. ~ Terada, ``Cryptanalysis of the Binary Permuted Kernel Problem,'' In International Conference on Applied Cryptography and Network Security, Springer, Cham, 2021,  pp. 396-423.

\bibitem{Persi2012} E. Persichetti, ``Compact McEliece keys based on quasi-dyadic Srivastava codes,'' Journal of Mathematical Cryptology, vol. 6, no. 2,
pp. 149–169, 2012.

\bibitem{Persichetti13} E. Persichetti, ``Secure and anonymous hybrid encryption from coding theory,'' in International Workshop on Post-Quantum
Cryptography. Springer, 2013, pp. 174–187.

\bibitem{Persichetti2018} E. Persichetti, ``Efficient one-time signatures from quasi-cyclic codes: A full treatment,'' Cryptography, vol. 2, no. 4, p. 30, 2018.

\bibitem{Persichetti12}E. Persichetti, ``Improving the efficiency of code-based cryptography,'' Ph.D. dissertation, University of Auckland, 2012.

\bibitem{Prange62} E. Prange, ``The use of information sets in decoding cyclic codes,'' IRE Transactions on Information Theory, vol. 8, no. 5, pp. 5–9,
1962. 




\bibitem{Sen11} N. Sendrier, ``Decoding one out of many,'' in International Workshop on Post-Quantum Cryptography. Springer, 2011, pp. 51–67.

\bibitem{SaSh18} S. Sato and J. Shikata, ``Lattice-based signcryption without random oracles,'' in International Conference on Post-Quantum Cryptography.
Springer, 2018, pp. 331–351.

\bibitem{SoLiLiLi17} Y. Song, Z. Li, Y. Li, and J. Li, ``Attribute-based signcryption scheme based on linear codes,'' Information Sciences, vol. 417, pp.
301–309, 2017.

\bibitem{SoHuMuWuHu20} Y. Song, X. Huang, Y. Mu, W. Wu, and H. Wang, ``A code-based signature scheme from the {L}yubashevsky framework,'' Theoretical
Computer Science, vol. 835, pp. 15–30, 2020.

\bibitem{StZh00} R. Steinfeld and Y. Zheng, ``A signcryption scheme based on integer factorization,'' in International Workshop on Information Security.
Springer, 2000, pp. 308–322.

\bibitem{Stern93} J. Stern, ``A new identification scheme based on syndrome decoding,'' in Annual International Cryptology Conference. Springer, 1993,
pp. 13–21.













\bibitem{YaCaLiXu19} X. Yang, H. Cao, W. Li, and H. Xuan, ``Improved lattice-based signcryption in the standard model,'' IEEE Access, vol. 7, pp. 155 552–155 562, 2019.

\bibitem{YaWaWaYaYa13} J. Yan, L. Wang, L. Wang, Y. Yang, W. Yao, ``Efficient lattice-based signcryption in standard model,'' Mathematical Problems in Engineering (2013).

\bibitem{YoFu07} M. Yoshida and T. Fujiwara, ``On the security of tag-$\kem$ for signcryption,'' Electronic Notes in Theoretical Computer Science, vol.
171, no. 1, pp. 83–91, 2007.

\bibitem{ZhWa14} X. Zhao, X. Wang, ``An efficient identity-based signcryption from lattice,'' International Journal of Security and Its Applications8(2), 363–374 (2014).

\bibitem{Zh97} Y. Zheng, ``Digital signcryption or how to achieve cost (signature \& encryption)<< cost (signature)+ cost (encryption),'' in Annual international cryptology conference. Springer, 1997, pp. 165–179.

\bibitem{ZhIm98} Y. Zheng and H. Imai, ``How to construct efficient signcryption schemes on elliptic curves,'' Information processing letters, vol. 68,
no. 5, pp. 227–233, 1998
\end{thebibliography}

\appendix
\section{PKE.Game}\label{pkegame}

Here we recall the $\indcca$ game for PKE called $\pkegame$ in our scheme. The decryption oracle is denoted by $\mathcal{O}$. 
\begin{figure}[ht]
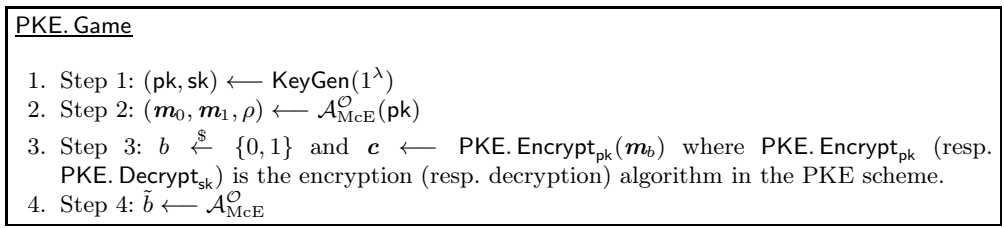

	\begin{center}
		\framebox{\parbox{13cm}{
\begin{minipage}[right]{13cm}

\noindent\textbf{\underline{$\pkegame$}}

\begin{enumerate} 
\item Step 1: $(\pk,\sk) \longleftarrow \keygen(1^{\lambda})$ 
\item Step 2: $(\textit{\textbf{m}}_0, \textit{\textbf{m}}_1, \rho) \longleftarrow \mathcal{A}_{\text{McE}}^{\mathcal{O}}(\pk)$
\item Step 3: $b\stackrel{\$}{\leftarrow} \{0,1\}$ and $\textit{\textbf{c}}\longleftarrow \pkeenc_{\pk}(\textit{\textbf{m}}_b)$ where $\pkeenc_{\pk}$ (resp. $\pkedec_{\sk}$) is the encryption (resp. decryption) algorithm in the PKE scheme.  
\item Step 4: $\tilde{b}\longleftarrow  \mathcal{A}_{\text{McE}}^{\mathcal{O}}$
\end{enumerate} 
\end{minipage}
}}  
		\end{center}
		\caption{$\pkegame$}
		\label{fig:pkegame}
	\end{figure}
	
In Step 4. the adversary $\mathcal{A}_{\text{McE}}$ is restricted not to make request to $\mathcal{O}$ on the ciphertext $\textit{\textbf{c}}$. Clear texts $\textit{\textbf{m}}_0$ and $\textit{\textbf{m}}_1$ must have the same length. $\mathcal{A}_{\text{McE}}$ wins when $\tilde{b}=b$ and its advantage corresponds to the probability that it wins this game which is denoted by 
$\epsilon_{pke}$.

\section{Target Preimage-Free}\label{PreImage}

Target Preimage-Free function is a special case of universal one-way function; An adversary is given $(\mathcal{H}, \textit{\textbf{y}})$ (chosen at random in their domain) and then attempts to find $\textit{\textbf{x}}$
such that $\mathcal{H}(\textit{\textbf{x}})=\textit{\textbf{y}}$. Let $\chi_{\lambda}=\{ X \}$ be a collection of domains and $\chi=\{ \chi_{\lambda}\}_{\lambda\in \mathbb{N}}$. Let $\tilde{\mathcal{H}}_{\lambda}=\{ \mathcal{H}:X \longrightarrow \{0,1\}^{\lambda}\ : \ X\in \chi_{\lambda} \}$ and $\tilde{\mathcal{H}}=\{ \tilde{\mathcal{H}}_{\lambda} \}_{\lambda\in \mathbb{N}}$. Note that $X$ is identified by the
description of $\mathcal{H}$. Let $\mathcal{A}_{\text{PreIm}}$ be an adversary playing the following game. 

\begin{figure}[ht]
\begin{center}
\framebox{\parbox{13cm}{
\begin{minipage}[right]{13cm}

\noindent\textbf{\underline{$\pregame$}}

\begin{enumerate} 
\item Step 1: $\mathcal{H}\longrightarrow \tilde{\mathcal{H}}_{\lambda}$
\item Step 2: $\textit{\textbf{y}}\longrightarrow \{0,1\}^{\lambda}$
\item Step 3: $\textit{\textbf{x}} \longrightarrow  \mathcal{A}_{\text{PreIm}}(\mathcal{H}, \textit{\textbf{y}})$ such that $\textit{\textbf{x}}\in X$.
\end{enumerate}
\end{minipage}
}}  
\end{center}
\caption{Preimage game}
\label{fig:PreIm}
\end{figure}

$\mathcal{A}_{\text{PreIm}}$ wins the game when $\mathcal{H}(\textit{\textbf{x}})=\textit{\textbf{y}}$ {and} the advantage of $\mathcal{A}_{\text{PreIm}}$ is the probability that it wins $\pregame$ for a given $\mathcal{H}\longrightarrow \tilde{\mathcal{H}}_{\lambda}$ and $\textit{\textbf{y}}\in \{0,1 \}^{\lambda}$. We say that $\tilde{\mathcal{H}}$ is Target Preimage free with regard to $\chi$ when the advantage $\epsilon_{\text{PreIm}}$ of $\mathcal{A}_{\text{PreIm}}$ is negligible.
\section{Security of the McEliece encryption with Fujisaki-Okamoto transformation}\label{McElie}

For the IND-CCA security of McEliece's scheme described in Figure \ref{fig:McElieceFO}, we need the following definition:

\begin{definition}\label{11}($\gamma$-uniformity \cite{CaHoPe12})
A public key encryption scheme $\Pi$ is called $\gamma$-uniform and $\mathcal{R}$ be the set where the randomness to be used in the (probabilistic) encryption is chosen. For a given key-pair
$(\pk, \sk)$, $\textit{\textbf{x}}$ be a plaintext and a string $\textit{\textbf{y}}$, we define
$$
\gamma(\textit{\textbf{y}})=Pr[\textit{\textbf{r}}\stackrel{\$}{\leftarrow}\mathcal{R}: \textit{\textbf{y}}=\mathcal{E}_{\pk}(\textit{\textbf{x}}, \textit{\textbf{r}})]
$$
where the notation $\mathcal{E}_{\pk}(\textit{\textbf{x}}, \textit{\textbf{r}})$ makes  the role of the randomness $\textit{\textbf{r}}$ {explicit}. We say that $\Pi$ is $\gamma$-uniform if, for any key-pair $(\pk, \sk)$, any plaintext $\textit{\textbf{x}}$  and any ciphertext $\textit{\textbf{y}}$, $\gamma (\textit{\textbf{x}}, \textit{\textbf{y}}) \leq  \gamma$ for a certain $\gamma \in \mathbb{R}$.
\end{definition}
We now can state the following lemma.

\begin{lemma}

The McEliece scheme with the Fujisaki-Okamoto transformation described in Figure \ref{fig:McElieceFO} is $\gamma$ uniform with
$$
\gamma=\frac{1}{2^{\tilde{k}}\binom{n}{t}}
$$
\end{lemma}

\begin{proof}
For any vector $\textit{\textbf{y}}\in \mathbb{F}_2^{n_r}$, either $\textit{\textbf{y}}$ is a word at distance $t$ from the code $\mathcal{C}$ of generator matrix $\mathbf{G}_{\pk,r}$, or it isn’t. When $\textit{\textbf{y}}$ is not a distance $t$ of $\mathcal{C}$, the probability for it to be a valid ciphertext is equal to 0. Else there is only one choice for $\textit{\textbf{r}}$ and $\textit{\textbf{e}}$ such that $\textit{\textbf{y}}=\textit{\textbf{r}}\mathbf{G}_{\pk,r}\oplus\textit{\textbf{e}}$, i.e.,
$$
\pr (d(\textit{\textbf{y}}, \mathcal{C}))=t)=\frac{1}{2^{\tilde{k}}\binom{n_r}{t}}
$$
\end{proof}
\begin{theorem}
{Under Assumptions 1, 3, and 5 the McEliece scheme based on a subcode of Goppa code with the Fujisaki-Okamoto transformation described in Figure \ref{fig:McElieceFO} is $\indcca$ secure.}
\end{theorem}

\begin{proof}
In Figure \ref{fig:McElieceFO}, the symmetric encryption used is the XOR function which is a one-time pad. Under Assumptions 1 and 3, the old McEliece encryption scheme is one-way secure. Therefore according to Theorem 12 of \cite{FuOk99}, the McEliece scheme with the Fujisaki-Okamoto transformation is $\indcca$ secure. 
\end{proof}
\end{document}